\documentclass[12pt,a4paper]{article} 

\usepackage[utf8]{inputenc}
\usepackage[english]{babel}
\usepackage[T1]{fontenc}
\usepackage{dirtytalk} 
\usepackage{ragged2e}
\usepackage{geometry}
\usepackage[table]{xcolor}
\definecolor{olive}{rgb}{0.3, 0.4, .1}
\definecolor{fore}{RGB}{249,242,215}
\definecolor{back}{RGB}{51,51,51}
\definecolor{title}{RGB}{255,0,90}
\definecolor{blackViolet}{RGB}{138,43,226}
\definecolor{gold}{rgb}{1.,0.84,0.}
\definecolor{JungleGreen}{cmyk}{0.99,0,0.52,0}
\definecolor{blackGreen}{cmyk}{0.85,0,0.33,0}
\definecolor{RawSienna}{cmyk}{0,0.72,1,0.45}
\definecolor{Magenta}{cmyk}{0,1,0,0}
\definecolor{wood}{RGB}{139,115,85}
\definecolor{dorange}{RGB}{255,127,0}
\definecolor{dolive}{RGB}{85,107,47}
\definecolor{drg}{RGB}{255,165,0}
\definecolor{dgreen}{rgb}{0,.5,0}
\definecolor{dblue}{rgb}{0,0,.5}
\definecolor{dred}{rgb}{0.5,0,.5}

\newcommand{\tcr}{\textcolor{red}}
\newcommand{\tcb}{\textcolor{blue}}

\usepackage{graphicx}% Include figure files
\usepackage{rotating}
\usepackage{tikz}
\usetikzlibrary{decorations.markings}

\usepackage{gnuplottex}
\usepackage{float}

\usepackage{xspace}

\usepackage{titlesec}
\titlespacing*{\section}{0pt}{2ex}{2ex}
\titlespacing*{\subsection}{0pt}{2ex}{2ex} 
\titlespacing*{\subsubsection}{0pt}{2ex}{2ex}

\titleformat*{\section}{\large\bfseries}
\titleformat*{\subsection}{\large\bfseries}
\titleformat*{\subsubsection}{\large\bfseries}
\titleformat*{\paragraph}{\large\bfseries}
\titleformat*{\subparagraph}{\large\bfseries}

\usepackage[colorlinks,citecolor=blue,urlcolor=blue,linkcolor=blue]{hyperref}
\usepackage{cite}

\usepackage{colortbl}
\usepackage{dcolumn}% Align table columns on decimal point
\usepackage{multirow}
\usepackage{longtable}
\usepackage{arydshln}
\usepackage{array}
\usepackage{tabularx, booktabs}
\newcolumntype{Y}{>{\centering\arraybackslash}X}

\usepackage{mathtools}
\usepackage{amsmath}
\usepackage{amsfonts}
\usepackage{amssymb} 
\usepackage{amsthm}
\usepackage{mathrsfs}
\usepackage{dsfont}
\usepackage{bm}% bold math
\usepackage{bbm}
\DeclareMathAlphabet{\mathpzc}{OT1}{pzc}{m}{it}

\usepackage{stmaryrd}

\makeatletter
\newcommand\xLongLeftRightArrow[2][]{%
  \ext@arrow 0099{\LongLeftRightArrowfill@}{#1}{#2}}
\def\LongLeftRightArrowfill@{%
  \arrowfill@\Leftarrow\Relbar\Rightarrow}
\makeatother

\usepackage{simplewick}
\makeatletter
\ExplSyntaxOn
\RenewDocumentCommand{\acontraction}{ O{1ex} O{black} m m m m}{%
  \tl_if_empty:nTF{#1}{\def\wickoffset{1ex}}{\def\wickoffset{#1}}%
  {\color{#2}%
  \mathchoice
    {\acontraction@\displaystyle{#3}{#4}{#5}{#6}{\wickoffset}}%
    {\acontraction@\textstyle{#3}{#4}{#5}{#6}{\wickoffset}}%
    {\acontraction@\scriptstyle{#3}{#4}{#5}{#6}{\wickoffset}}%
    {\acontraction@\scriptscriptstyle{#3}{#4}{#5}{#6}{\wickoffset}}}}%
\RenewDocumentCommand{\bcontraction}{ O{1ex} O{black} m m m m}{%
  \tl_if_empty:nTF{#1}{\def\wickoffset{1ex}}{\def\wickoffset{#1}}%
  {\color{#2}%
  \mathchoice
    {\bcontraction@\displaystyle{#3}{#4}{#5}{#6}{\wickoffset}}%
    {\bcontraction@\textstyle{#3}{#4}{#5}{#6}{\wickoffset}}%
    {\bcontraction@\scriptstyle{#3}{#4}{#5}{#6}{\wickoffset}}%
    {\bcontraction@\scriptscriptstyle{#3}{#4}{#5}{#6}{\wickoffset}}}}%
\ExplSyntaxOff
\makeatletter

\newcommand\reallywidehat[1]{%
\savestack{\tmpbox}{\stretchto{%
  \scaleto{%
    \scalerel*[\widthof{\ensuremath{#1}}]{\kern-.6pt\bigwedge\kern-.6pt}%
    {\rule[-\textdepth/2]{1ex}{\textdepth}}%WIDTH-LIMITED BIG WEDGE
  }{\textdepth}% 
}{0.5ex}}%
\stackon[1pt]{#1}{\tmpbox}%
}

\usepackage{physics} 
\usepackage{braket}

\newtheorem{theorem}{Theorem}[section]
\newtheorem{lemma}{Lemma}[section]

\newtheorem{proposition}{Proposition}[section]
\newtheorem{definition}{Definition}[section]

\newtheorem{example}{Example}[section]
\newtheorem{remark}{Remark}[section]
\newtheorem{ourrule}{Rule}[section]

\usepackage{algorithm}
\usepackage[noend]{algpseudocode}
\usepackage{listings}

\if False
\documentclass[multi=my,crop]{standalone}

\usepackage{mathtools}
\usepackage[utf8]{inputenc}
\usepackage[french]{babel}
\usepackage[T1]{fontenc}
\usepackage{amsmath}
\usepackage{amsfonts}
\usepackage{bm}
\usepackage{stmaryrd}
\usepackage{amssymb}
\usepackage{mathrsfs}
\usepackage{bm}
\usepackage{graphicx}
\usepackage{ragged2e}
\usepackage{fancybox}
\usepackage{bbm}

\usepackage{tikz}
\usetikzlibrary{backgrounds}

\pgfdeclarelayer{bg}    % declare background layer
\pgfsetlayers{bg,main}
\usepackage{colortbl}
\usepackage{multimedia}
\usepackage{xcolor}
\usepackage{braket}

\usepackage{amsfonts}
\usepackage{ifthen}

\fi

\newcommand{\pdeex}[3]{%excitation parenthèse
	\def\h{.5}
	\def\x{#1}
	\def\r{.5}
	\draw[line width=1pt]  
		(\x+\r, \h+\r)node[above] {$\hat{o}^\dag_{#2}$}  
		arc (90:180:\r) 
		-- (\x, -\h) 
		arc (180:270:\r) node[below] {$\hat{v}_{#3}$};
}

\newcommand{\pexci}[3]{%excitation parenthèse
	\def\h{.5}
	\def\x{#1}
	\def\r{.5}
	\draw[line width=1pt] 
		(\x-\r, \h+\r) node[above] {$\hat{o}_{#2}$}  
		arc (90:0:\r) 
		-- (\x, -\h) 
		arc (0:-90:\r) node[below] {$\hat{v}^\dag_{#3}$};
}

\newcommand{\padaga}[5]{
	\def\x{#1}
	\def\e{.015}%demi de l'épaisseur du trait
	\def\r{.5}
	
	\ifthenelse{\equal{\detokenize{#2}}{\detokenize{occ}}}
		{ \draw[fill] (\x,-\e) arc (270:90:\e+\r) node[above] {$\hat{a}^\dag_{#3}$}
			-- (\x,1+\e) arc (90:270:\r) 
			-- cycle; }
		{ \ifthenelse{\equal{\detokenize{#2}}{\detokenize{vir}}}	
			{ \draw[fill] (\x,\e)
				-- (\x-2*\e,\e)
				-- (\x-2*\e,\e)
				-- (\x-2*\e,-\e) arc (0:-90:2*\r) node[below] {$\hat{a}^\dag_{#3}$}
				-- (\x-2*\r,-2*\r-\e) arc (-90:0:2*\r) --cycle;}
			{ \draw (\x-1,0) node[above] {$\hat{a}^\dag_{#3}$};
			  \def\x{#1-1}}
		}
	
	\draw[fill] (\x,\e) 
		-- (\x+2,\e) 
		-- (\x+2,-\e) 
		-- (\x,-\e) 
		-- cycle;
	
	\ifthenelse{\equal{\detokenize{#4}}{\detokenize{occ}}}
		{ \draw[fill] (\x+2,\e) arc (-90:90:\r)  node[above] {$\hat{a}_{#5}$}
			-- (\x+2,1+\e) arc (90:-90:\e+\r) 
			-- cycle; ;}
		{\ifthenelse{\equal{\detokenize{#4}}{\detokenize{vir}}}
			{\draw[fill] (\x+2,\e) 
			-- (\x+2*\e+2,\e) 
			-- (\x+2*\e+2,\e)
			-- (\x+2*\e+2,-\e) arc (180:270:2*\r) node[below] {$\hat{a}_{#5}$}
			-- (\x+2*\r+2,-2*\r-\e) arc (270:180:2*\r) --cycle;}
			{\draw (\x+2,0) node[above] {$\hat{a}_{#5}$};}
		}	
}

\newcommand{\adaga}[4]{%one-body operator
  \ifthenelse{\equal{\detokenize{#1}}{\detokenize{exci}}}
   { \draw[line width=1pt] (0,-1) arc(-90:0:1);
    \tikzset{shift={(1,0)}} }{}
  \ifthenelse{\equal{\detokenize{#1}}{\detokenize{deex}}}
   { \draw[line width=1pt] (0,1) arc(90:270:.5) -- (1,0);
    \tikzset{shift={(1,0)}} }{}

   \draw[line width=1pt] (0,0)node[above]{$\hat{a}^\dag_{#2}$} -- (1.5,0)node[above]{$\hat{a}_{#4}$};
   \tikzset{shift={(1.5,0)}}

  \ifthenelse{\equal{\detokenize{#3}}{\detokenize{deex}}}
   { \draw[line width=1pt] (0,0) arc(180:270:1);
    \tikzset{shift={(1,0)}} }{}
  \ifthenelse{\equal{\detokenize{#3}}{\detokenize{exci}}}
   { \draw[line width=1pt] (0,0) -- (1,0) arc(-90:90:.5);
    \tikzset{shift={(1,0)}} }{}

}

\newcommand{\exci}[2]{%excitation
   \draw[line width=1pt] 
 	     (0,-1) node[below right]{$\hat{v}^\dag_{#2}$} arc (-90:0:1.5) 
 	  -- (1.5,.5) arc (0:90:.5) node[above right]{$\hat{o}_{#1}$};
   \tikzset{shift={(1,0)}}
}

\newcommand{\deex}[2]{%deexcitation
   \draw[line width=1pt] 
   		(0,1) node[above left]{$\hat{o}^\dag_{#2}$} arc (90:180:.5) 
   	  --(-.5,.5) arc (180:270:1.5)node[below left]{$\hat{v}_{#1}$};
   \tikzset{shift={(1,0)}}
}

\newcommand{\exciend}[2]{%excitation
   \draw[line width=1pt] 
 	     (0,-1) node[below right]{$\hat{v}^\dag_{#2}$} arc (-90:0:.5) 
 	  -- (.5,.5) arc (0:90:.5) node[above right]{$\hat{o}_{#1}$};

}

\newcommand{\deexend}[2]{%deexcitation
   \draw[line width=1pt] 
   		(0,1) node[above left]{$\hat{o}^\dag_{#2}$} arc (90:180:.5) 
   	  --(-.5,-.5) arc (180:270:.5)node[below left]{$\hat{v}_{#1}$};
}

\if False
\usepackage{pgfplots}
\fi

\usetikzlibrary{calc}

\makeatletter
\newcommand\footnoteref[1]{\protected@xdef\@thefnmark{\ref{#1}}\@footnotemark}
\makeatother

\begin{document}
\title{\normalsize{\textbf{Dyck language and fermionic second quantization}} \\ \normalsize{\textbf{II. Applications}}}

\date{} 

\maketitle

\vspace*{-2cm}

\noindent \begin{center}
\textit{by} Jérémy Morere\footnote{\label{note1}\textit{Université de Lorraine,} CNRS, LPCT, \textit{F-54000 Nancy, France}}{}$^,$\footnote{jeremy.morere@univ-lorraine.fr} \textit{and} Thibaud Etienne\footnoteref{note1}{}$^,$\footnote{thibaud.etienne@univ-lorraine.fr}\\
\end{center}

$\;$

\begin{abstract}
\noindent In this work, we establish a direct connection between supplemented Dyck language and the signed expectation value of chains of second quantization operators relatively to the physical vacuum and relatively to a one-determinant state.
Inspired by the fact that Dyck language provides an example of the emergence of the Catalan numbers in linguistic framework analysis, we show that these numbers are central when numbering the terms remaining when eliminating vanishing contributions detected by our application of Dyck language to fermionic second quantization.
From the translation of creation and annihilation operator — or of pairs of operators — into a bracket alphabet, we derive simple and intuitive sufficient conditions for the nullity of expectation values that does not require an explicit application of Wick’s theorem. This is done here with respect to the physical vacuum or relatively to a one-determinant state.
We also extend this translation into a diagrammatic framework that allows a visual determination of the signature of fully contracted terms, reproducing the results of Wick’s theorem. This approach has been extended to the case of (nested) commutators of pairs of fermionic second quantization including at least one excitation or deexcitation operator. Our results have been implemented in a software, \texttt{MobiDyck}, whose source code is freely available on the web. The algorithmic approach inspired by our work on Dyck language is detailed in this paper. Finally, a comparison of our diagrammatic approach with Goldstone diagrams is provided and closes the article.
\\ $\;$ 
\end{abstract}

\section{Introduction}

Over the past decades, coupled-cluster \cite{helgaker_molecular_2014, lefebvre_use_1969, cizek_correlation_1966, cizek_correlation_1971, shavitt_many-body_2009, paldus_correlation_1972} theory has emerged as one of the most accurate and widely used methods for treating electron correlation. Its success stems from its size extensivity and its systematic improvability through hierarchical truncation schemes. The structure of coupled-cluster equations involves exponential operator expansions, which are expressed as sums of nested commutators.

Conventional coupled cluster (CC) methods rely on a truncation based on excitation rank. For a given truncation, the development of the exponential can be truncated as well. The maximal number of nested commutators and their composition are known in advance. It is possible to pre-compute which term will have a non-zero expectation value, avoiding the development of nested commutators and their on-the-fly evaluation for each CC calculation.

Other approaches, like adaptive-CC \cite{evangelista_adaptive_2014, lyakh_adaptive_2010}, introduce an adaptive truncation. All cluster amplitudes of a given excitation rank are not taken into account, the set of these amplitudes can be updated at each iteration. A direct consequence is the impossibility to pre-compute the expectation values in advance. Similarly, at each iteration nested commutators have to be developed and expectation values of chains of fermionic second quantization operators have to be computed.

In the general case — i.e., not necessarily relatively to a vacuum state —, different approaches can be chosen in order to evaluate the expectation value for these kinds of chains of operators. Some direct methods can be used for this purpose, like applying all the second quantization operators one by one, or repeatedly applying anti-commutation rules to reorder the operators. Those methods require to carefully track the sign of each term and, in the relevant cases, duly discarding vanishing terms.
For long chains of operators, these approaches generate a tremendous number of intermediate expressions, increasing the risk of introducing errors in the development when performed by hand. An alternative approach consists in applying the time-independent Wick theorem for fermions \cite{wick_evaluation_1950,lindgren_atomic_1986,surjan_second_1989,gross_many-particle_1991,kutzelnigg_normal_1997,wilson_methods_1992,shavitt_many-body_2009}, which relies on two operations: contracting two operators and normal-ordering a chain — \textit{vide infra}.

Despite its strength, in the context of expectation value evaluation, Wick’s theorem introduces a new layer of combinatorial complexity without giving an explicit procedure for contracting all operators. The number of possible contractions grows factorially with the number of operators, and many of these contributions ultimately vanish.

In Part I of our contribution, we introduced a one-dimensional representation of any chain of second quantization operators based on Dyck language. According to the reference state at stake -- physical or Fermi vacuum -- different ``translations'' of second quantization operators into a bracket-based alphabet have been defined, each revealing by simple syntactic criteria whether the expectation value of the chain is necessarily zero or not, bypassing explicit operator manipulation. 
However, that framework was limited to nullity criteria and did not give access to the explicit value -- zero, one, or minus one -- of the expectation value. In this second part, we extend this approach by supplementing the one-dimensional representation toward a two-dimensional diagrammatic representation, in order to recover results obtained from the corollary of Wick's theorem.
The use of two-dimensional representations of problems involving second quantization is common. We can cite Goldstone\cite{goldstone_derivation_1957,mattuck_guide_1992,szabo_modern_2012} and Hugenholtz\cite{hugenholtz_perturbation_1957,mattuck_guide_1992,szabo_modern_2012} diagrams for many-body perturbation theory or coupled-cluster diagrams \cite{shavitt_many-body_2009,kucharski_fifth-order_1986,crawford_introduction_2000} for CC. Each diagram represents, for a dedicated problem, a chain of second quantization operators.

In Section II, Wick's theorem is introduced and expressed in the specific context of expectation value evaluation. The Fermi-$\mathcal{L}_1$ translation, introduced in Part I, is supplemented in order to recover results from the corollary of Wick's theorem from a diagrammatic representation.

Section III presents combinatorial, numbering results. Descriptors coming from Dyck language are introduced into the second quantization formalism in order to count sequences related to chains of second quantization operators.

The simplification of some types of (nested) commutators is studied in Section IV. Manipulation rules specific to the simplification of those (nested) commutators are added to the two-dimensional diagrammatic representation introduced in Section II. These rules describe how to diagrammatically simplify some (nested) commutators.

Section V exposes an algorithm inspired by Dyck words in order to symbolically evaluate the expectation value of chains of second quantization operators, and section VI establishes a link between our two-dimensional diagrammatic representations and Goldstone diagrams.

\section{Time-independent Wick theorem for fermions}

\noindent Time-independent Wick theorem for fermions --- or in short Wick's theorem from now on --- provides a systematic way for evaluating expectation values of chains of second quantization operators. The theorem relies on two operations: contracting two operators and normal-ordering a chain. 

In this study we will focus on a corollary of Wick's theorem rather than Wick's theorem itself. For this reason, the two aformentioned operations do not need to be explicitly introduced here — we will solely need to \textit{use} contractions in dedicated cases and relatively to a reference state; the notion of ``normal ordering'' of a chain, which is central in the original formulation of Wick's theorem, has no use in the context of the corollary of interest here.

\subsection{Expectation value of contracted chains of operators}

Let $\Psi$ be a one-determinant quantum state written solely with elements of $\mathcal{B}$. The expectation value of a contraction between a creation operator and an annihilation operator (written in that order), reads
\begin{equation}\label{eq:contraction_occ}
\forall (r,s) \in \llbracket 1,L\rrbracket^2, \;
\braket{\acontraction[][black]{}{\hat{a}}{^\dag_{r}}{\hat{a}}\hat{a}^\dag_{r}\hat{a}_{s}}_{\Psi} 
= \delta_{r, s}n_r^{\Psi} n_s^{\Psi}.
\end{equation}
\noindent If the operator order is reversed, we obtain
\begin{equation}\label{eq:contraction_vir}
\forall (r,s) \in \llbracket 1,L\rrbracket^2, \;
\braket{\acontraction[][black]{}{\hat{a}}{_{s}}{\hat{a}}\hat{a}_{s}\hat{a}^\dag_{r}}_{\Psi} 
= \delta_{r, s}(1-n_r^{\Psi} )(1-n_s^{\Psi} ).
\end{equation}
Notice that the $\Psi$-expectation value of contractions between two creation operators or between two annihilation operators is equal to zero. Notice also that contractions are not limited to two-operator chains: a chain of $n$ second quantization operators can be inserted between the contracted ones. For example, consider $\hat{a}_r^\dag\hat{a}_s$ and $\hat{a}_s\hat{a}_r^\dag$ seen in \eqref{eq:contraction_occ} and \eqref{eq:contraction_vir}, respectively. Consider any $n$-tuple of second quantization operators $(\hat{p}_i)_{1\leq i \leq n}$. We have

\begin{equation}\label{eq:contraction_in_chain-occ}
\forall (r,s) \in \llbracket 1,L\rrbracket^2, \;
\braket{\acontraction[][black]
{\cdots}
{\hat{a}}
{^\dag_{r}\hat{p}_1\cdots\hat{p}_n}
{\hat{a}}
\cdots\hat{a}^\dag_{r}\hat{p}_1\cdots\hat{p}_n\hat{a}_{s}\cdots}_{\Psi} 
= (-1)^n\braket{
\acontraction[][black]{}{\hat{a}}{^\dag_{r}}{\hat{a}}\hat{a}^\dag_{r}\hat{a}_{s}}_{\Psi}  
\braket{\cdots\hat{p}_1\cdots\hat{p}_n\cdots
}_{\Psi}.
\end{equation}
\noindent The same relation stands if the operator order is reversed:
\begin{equation}\label{eq:contraction_in_chain-vir}
\forall (r,s) \in \llbracket 1,L\rrbracket^2, \;
\braket{\acontraction[][black]
{\cdots}
{\hat{a}}
{_{s}\hat{p}_1\cdots\hat{p}_n}
{\hat{a}}
\cdots\hat{a}_{s}\hat{p}_1\cdots\hat{p}_n\hat{a}^\dag_{r}\cdots}_{\Psi} 
= (-1)^n\braket{
\acontraction[][black]{}{\hat{a}}{_{s}}{\hat{a}}\hat{a}_{s}\hat{a}^\dag_{r}}_{\Psi}  
\braket{\cdots\hat{p}_1\cdots\hat{p}_n\cdots
}_{\Psi}.
\end{equation}
\noindent Within a chain, multiple pairs of operators can be contracted simultaneously. In order to illustrate this, we now work in the Fermi vacuum for the remaining of this section. The following example explores the evaluation of the expectation value of a fully contracted chain — each operator of the chain is contracted with another one —, and shows that it can be written as a product of Kronecker deltas.

\begin{example}\label{ex:fully_contracted}
    Let $(i,j)$ be a couple of integers, both belonging to $\llbracket 1,N\rrbracket$. Let $(a,b,c,d)$ be a 4-tuple of integers, all belonging to $\llbracket N+1,L\rrbracket$. Applying formulas \eqref{eq:contraction_occ} to \eqref{eq:contraction_in_chain-vir}, we see that 
\begin{equation*}
\braket{
\acontraction[1ex]{}{\hat{o}}{^\dag_i\hat{v}_a\hat{v}^\dag_b\hat{v}_c\hat{v}^\dag_d}{\hat{o}}
\acontraction[2ex]{\hat{o}^\dag_i}{\hat{v}}{_a}{\hat{v}}
\acontraction[2ex]{\hat{o}^\dag_i\hat{v}_a\hat{v}^\dag_b}{\hat{v}}{_c}{\hat{v}}
\hat{o}^\dag_i\hat{v}_a\hat{v}^\dag_b\hat{v}_c\hat{v}^\dag_d\hat{o}_j
}_{\Psi_0}
=
\braket{\acontraction[][black]{}{\hat{o}}{^\dag_{i}}{\hat{o}}\hat{o}^\dag_{i}
\hat{o}_{j}}_{\Psi_0} 
\braket{
\acontraction[1ex]{}{\hat{v}}{_a}{\hat{v}}
\acontraction[1ex]{\hat{v}_a\hat{v}^\dag_b}{\hat{v}}{_c}{\hat{v}}
\hat{v}_a\hat{v}^\dag_b\hat{v}_c\hat{v}^\dag_d}_{\Psi_0} 
=
\braket{\acontraction[][black]{}{\hat{o}}{^\dag_{i}}{\hat{o}}\hat{o}^\dag_{i}
\hat{o}_{j}}_{\Psi_0} 
\braket{
\acontraction[1ex]{}{\hat{v}}{_a}{\hat{v}}
\hat{v}_a\hat{v}^\dag_d}_{\Psi_0} 
\braket{
\acontraction[1ex]{}{\hat{v}}{_c}{\hat{v}}
\hat{v}^\dag_b\hat{v}_c}_{\Psi_0} 
=\delta_{i,j}\delta_{a,b}\delta_{c,d}.
\end{equation*}
\end{example}

\noindent For a given $\hat{C}$ chain several full contractions of the chain are possible. Some of the full contractions of the chain have a zero expectation value.
\begin{example}
In the framework of {\normalfont Example \ref{ex:fully_contracted}}, we can exhibit the expectation value of another full contraction of the chain:
\begin{equation*}
\braket{
\acontraction[1ex]{}{\hat{o}}{^\dag_i}{\hat{v}}
\acontraction[1ex]{\hat{o}^\dag_i\hat{v}_a}{\hat{v}}{^\dag_b}{\hat{v}}
\acontraction[1ex]{\hat{o}^\dag_i\hat{v}_a\hat{v}^\dag_b\hat{v}_c}{\hat{v}}{^\dag_d}{\hat{o}}
\hat{o}^\dag_i\hat{v}_a\hat{v}^\dag_b\hat{v}_c\hat{v}^\dag_d\hat{o}_j
}_{\Psi_0}
= 0 
\end{equation*}
\end{example}
\noindent Contraction defined in \eqref{eq:contraction_occ} relatively to Fermi vacuum gives a non-zero expectation value only if both operators are \textit{o-}operators, while \eqref{eq:contraction_vir} in the Fermi vacuum is non-zero if both operators are \textit{v-}operators.

Notice that these considerations extend to the physical vacuum, in which case all the occupation numbers are simply equal to zero.

\subsection{Diagrammatic evaluation of expectation values}

\subsubsection{A corollary of Wick's theorem}\label{subsub:corwick}

\noindent {We are now fully equipped for bringing the corollary of Wick's theorem that will be of interest to us in this paper.}

\begin{theorem}[Corollary of the Wick's theorem]\label{theo:Wick}
The expectation value of a chain of second-quantization operators relatively to the physical vacuum or relatively to a one-determinant reference state can be expressed as a sum over all possible full contractions of the chain.
\end{theorem}
\noindent Example \ref{ex:WickFermion} shows how the expectation value of a chain of four second quantization operators can be decomposed into a sum of expectation values where the chain is fully contracted.
\begin{example}\label{ex:WickFermion}
Let $(i,j)$ be a couple of integers, both belonging to $\llbracket 1,N\rrbracket$. Let $(a,b)$ be a couple of integers, both belonging to $\llbracket N+1,L\rrbracket$. Applying {\normalfont{Theorem \ref{theo:Wick}}} to $\braket{\hat{o}^\dag_i\hat{v}_a\hat{v}^\dag_b\hat{o}_j
}_{\Psi_0}$ one finds
\begin{equation*}
\braket{\hat{o}^\dag_i\hat{v}_a\hat{v}^\dag_b\hat{o}_j
}_{\Psi_0}
=\braket{
\acontraction[1ex]{}{\hat{o}}{^\dag_i\hat{v}_a\hat{v}^\dag_b}{\hat{o}}
\acontraction[2ex]{\hat{o}^\dag_i}{\hat{v}}{_a}{\hat{v}}
\hat{o}^\dag_i\hat{v}_a\hat{v}^\dag_b\hat{o}_j
}_{\Psi_0}
+\braket{
\acontraction[1ex]{}{\hat{o}}{^\dag_i\hat{v}_a}{\hat{v}}
\acontraction[2ex]{\hat{o}^\dag_i}{\hat{v}}{_a\hat{v}^\dag_b}{\hat{o}}
\hat{o}^\dag_i\hat{v}_a\hat{v}^\dag_b\hat{o}_j
}_{\Psi_0}
+\braket{
\acontraction[1ex]{}{\hat{o}}{^\dag_i}{\hat{v}}
\acontraction[1ex]{\hat{o}^\dag_i\hat{v}_a}{\hat{v}}{^\dag_b}{\hat{o}}
\hat{o}^\dag_i\hat{v}_a\hat{v}^\dag_b\hat{o}_j
}_{\Psi_0}.
\end{equation*}
The second term is equal to zero because one contraction is performed over two creation operators and the other is performed over two annihilation operators. The third term is equal to zero in virtue of \eqref{eq:contraction_occ} and \eqref{eq:contraction_vir}. Using relations \eqref{eq:contraction_occ} to \eqref{eq:contraction_in_chain-vir} we conclude that
\begin{equation*}
\braket{\hat{o}^\dag_i\hat{v}_a\hat{v}^\dag_b\hat{o}_j
}_{\Psi_0}
=\braket{
\acontraction[1ex]{}{\hat{o}}{^\dag_i\hat{v}_a\hat{v}^\dag_b}{\hat{o}}
\acontraction[2ex]{\hat{o}^\dag_i}{\hat{v}}{_a}{\hat{v}}
\hat{o}^\dag_i\hat{v}_a\hat{v}^\dag_b\hat{o}_j
}_{\Psi_0}
=\delta_{i,j}\delta_{a,b}.
\end{equation*}
\end{example}

\noindent Formulas \eqref{eq:contraction_in_chain-occ} and \eqref{eq:contraction_in_chain-vir} introduce a signature — i.e., a ``+'' or a ``$-$'' sign —, which can be recovered in the case of fully contracted chains in two ways: (i) by applying successively relations  \eqref{eq:contraction_in_chain-occ} and \eqref{eq:contraction_in_chain-vir} themselves, or (ii) by using the graphical representation of the contractions and counting the crossings between the horizontal and vertical lines above the chain. If $\mathcal{N}$ is the number of crossings, the corresponding signature is equal to $(-1)^{\mathcal{N}}$.

\subsubsection{Diagrammatic supplementation of the $\mathcal{L}_{1,1}$ alphabet}
When evaluating the expectation value of a chain of second quantization operators relatively to a reference state — physical or Fermi vacuum —, Theorem \ref{theo:Wick} allows one to restrict him/herself to full contractions. The Dyck language representation can be used for the purpose of systematically excluding the fully contracted terms whose expectation value is zero. 

\begin{remark}
For the sake of brevity, in the following the Fermi-$\mathcal{L}_{1,1}$ translation introduced in {\normalfont Part I} will simply be termed ``translation''.
\end{remark}

\noindent We now introduce two construction rules for supplementing the $\mathcal{L}_{1,1}$ alphabet:
\begin{ourrule}[Brackets annotation]\label{def:bracket_annotation}
The upper extremity of the bracket resulting from the translation of a (de)excitation operator is annotated with its \textit{o-}operator, and the lower extremity with its \textit{v-}operator. An imaginary line passes through the middle of all the brackets, it splits the diagram into two half-planes. The upper half-plane contains extremities annotated with \textit{o-}operators and the lower half-plane contains extremities annotated \textit{v-}operators.
\end{ourrule}

\begin{ourrule}[Dashes annotation]\label{def:dash_annotation}
The left extremity of the dash resulting from the translation of an arbitrary creation-annihilation pair is annotated with its creation operator and the right extremity with its annihilation operator. The dash lies on the imaginary line mentioned in {\normalfont Rule \ref{def:bracket_annotation}}; its extremities belong to both the half-planes.
\end{ourrule}

\noindent The imaginary line mentioned in rules \ref{def:bracket_annotation} and \ref{def:dash_annotation} is a visual support to the reading of the diagrams, and does not contain any additional information regarding the diagrammatic representation. For this reason, this line will be drawn in gray in the diagrams below.

\begin{definition}[$\mathcal{L}_{1,1}$--diagram representation]
The $\mathcal{L}_{1,1}$--diagram representation of a chain containing (de)excitation operators with or without arbitrary creation-annihilation operator pair is the annotated translation of the chain.
\end{definition}

\begin{example}
Let $(i,j)$ be a couple of integers, both belonging to $\llbracket 1,N\rrbracket$. 
Let $(a,b)$ be a couple of integers, both belonging to $\llbracket N+1,L\rrbracket$. 
Let $\hat{a}^\dag_r$ and $\hat{a}_s$ be two second quantization operators, each pointing at a spin-orbital with non-definite occupation number relatively to the Fermi vacuum. The diagrammatic representation of the deexcitation operator $\hat{D}_a^i$, of the excitation operator $\hat{E}_j^b$, and of the arbitrary pair of creation and annihilation $\hat{a}^\dag_r\hat{a}_s$ are illustrated in {\normalfont Figure \ref{fig:deex-adaga-exci}}.
\end{example}

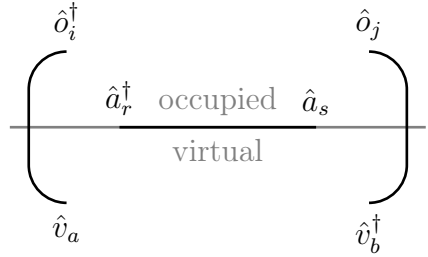
\begin{figure}[h!]
    \centering
    \begin{tikzpicture}
%\draw[grid] (-1,2) grid (6,-16);

\draw[color=gray, line width=1pt] (-.25,0) -- (5.25,0) node[midway,above] {occupied} node[midway,below] {virtual};
%%\draw (-.5,1.25) node {a)};
%%
\pdeex{0}{i}{a}
\draw[line width=1pt] (1.2,0)node[above]{$\hat{a}^\dag_r$} -- (3.8,0)node[above]{$\hat{a}_s$};
\pexci{5}{j}{b}
%\draw (2.5,-1.3) node {$\hat{D}_a^i\hat{a}^\dag_r\hat{a}_s\hat{E}_j^b$};

\end{tikzpicture}
    \vspace*{-.5cm}
    \caption{$\mathcal{L}_{1,1}$--diagram representation of the $\hat{D}_a^i\hat{a}^\dag_r\hat{a}_s\hat{E}_j^b$ chain.}
    \label{fig:deex-adaga-exci}
\end{figure}
\noindent We now provide two construction rules in view of connecting the extremities of $\mathcal{L}_{1,1}$ diagrams. 
\begin{ourrule}[Directionality of bracket extremities]
The extremities of an opening (respectively, closing) bracket are said to ``point'' in the right (respectively, left) direction.
\end{ourrule}
\noindent This rule is illustrated in Figure \ref{fig:direction_brackets} with the example of couple of one excitation and one deexcitation operators represented using annotated opening and closing brackets. 
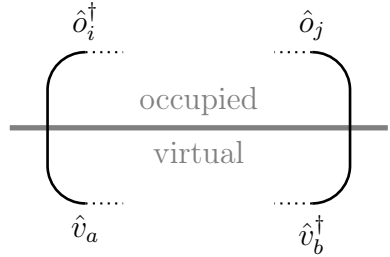
\begin{figure}[h!]
\begin{center}
\begin{tikzpicture}
    \draw[line width=2pt, gray] (-.5,0) -- (4.5,0) node[midway,above]{occupied} node[midway,below] {virtual}; 
    \pdeex{0}{i}{a}\pexci{4}{j}{b}
    
    \draw[dotted,thick] (.5,1) -- ++(.5,0);
    \draw[dotted,thick] (.5,-1) -- ++(.5,0);
    \draw[dotted,thick] (3.5,1) -- ++(-.5,0);
    \draw[dotted,thick] (3.5,-1) -- ++(-.5,0);
\end{tikzpicture}
\vspace*{-1cm}
\end{center}
\caption{Illustration of the assignment of a directionality to annotated opening and closing brackets extremities: the dotted lines highlight the fact that we assign a left directionality to extremities of closing brackets, and a right directionality to extremities of opening brackets.}
\label{fig:direction_brackets}
\end{figure}

\begin{ourrule}[Bending extremities of a dash]
The extremities of a dash can be bent toward the upper or lower half-plane, pointing in left or right direction.
\end{ourrule}
\noindent This rule is illustrated with four examples of directionality assignment to the two extremities of a given dash in Figure \ref{fig:adaga_graph}.
\begin{figure}[h!]
\begin{center}
\begin{tikzpicture}
%\draw[grid] (0,-2) grid (4,2);
%\draw[line width=2pt, gray] (-.5,0) -- (4.5,0) node[midway,above]{occupied} node[midway,below]{virtual}; 
%\padaga{2}{}{r}{}{s}
%\draw (-.2,0) node[above]{a)};
%\tikzset{shift={(0,-2)}}
\draw[line width=2pt, gray] (-.5,0) -- (4.5,0) node[midway,above]{occupied} node[midway,below]{virtual}; 
\padaga{1}{occ}{r}{occ}{s}
\draw (-.2,0) node[above]{a)};

\draw[dotted,thick] (1,1.02) -- ++(.5,0); 
\draw[dotted,thick] (3,1.02) -- ++(-.5,0); 

\tikzset{shift={(5,0)}}
\draw[line width=2pt, gray] (-.5,0) -- (4.5,0) node[midway,above]{occupied} node[midway,below]{virtual}; 
\padaga{1}{occ}{r}{vir}{s}
\draw (-.2,0) node[above]{b)};

\draw[dotted,thick] (1,1.02) -- ++(.5,0); 
\draw[dotted,thick] (4,-1.02) -- ++(.5,0); 

\tikzset{shift={(-5,-3)}}
\draw[line width=2pt, gray] (-.5,0) -- (4.5,0) node[midway,above]{occupied} node[midway,below]{virtual}; 
\padaga{1}{vir}{r}{occ}{s}
\draw (-.2,0) node[above]{c)};

\draw[dotted,thick] (0,-1.02) -- ++(-.5,0); 
\draw[dotted,thick] (3,1.02) -- ++(-.5,0); 

\tikzset{shift={(5,0)}}
\draw[line width=2pt, gray] (-.5,0) -- (4.5,0) node[midway,above]{occupied} node[midway,below]{virtual}; 
\padaga{1}{vir}{s}{vir}{s}
\draw (-.2,0) node[above]{d)};

\draw[dotted,thick] (0,-1.02) -- ++(-.5,0); 
\draw[dotted,thick] (4,-1.02) -- ++(.5,0); 

\end{tikzpicture}
\end{center}
\vspace*{-.8cm}
\caption{Four examples of ways of bending the two extremities of a dash: a) top-top/right-left bending; b) top-bottom/right-right bending; c) bottom-top/left-left bending; d) bottom-bottom/left-right bending.}
\label{fig:adaga_graph}
\end{figure}
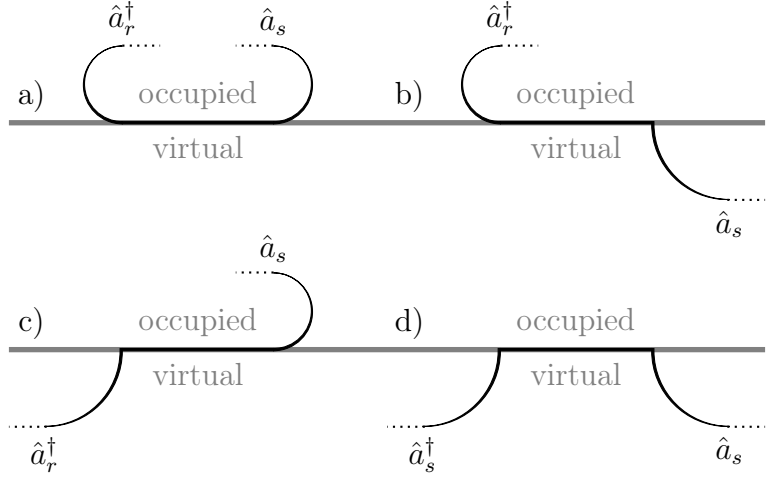
\begin{definition}[Connected $\mathcal{L}_{1,1}$ diagram]
An $\mathcal{L}_{1,1}$ diagram is said to be {\normalfont connected} when at least two of its extremities are linked using a line.
\end{definition}

\begin{definition}[Fully connected $\mathcal{L}_{1,1}$ diagram]
An $\mathcal{L}_{1,1}$ diagram is said to be {\normalfont fully connected} when each of its extremities is linked to another one using a line.
\end{definition}

\noindent From now on, fully connected $\mathcal{L}_{1,1}$ diagrams will simply be termed ``fully connected diagrams''.

\begin{example}\label{ex:one_connection}
Let $(i,j)$ be a couple of integers, both belonging to $\llbracket 1,N\rrbracket$. 
Let $(a,b)$ be a couple of integers, both belonging to $\llbracket N+1,L\rrbracket$.
Consider the following chain: $\hat{D}_a^i\hat{E}_j^b = \hat{o}_i^\dag\hat{v}_a\hat{v}_b^\dag\hat{o}_j$. There are only three ways to fully connect its $\mathcal{L}_{1,1}$ diagram. They are all accounted for in Figure \ref{fig:one_connection}.
\begin{figure}[h!]
\begin{center}
\begin{tikzpicture}
    \draw[line width=2pt, gray] (-.5,0) -- (4.5,0) node[midway,above]{occupied} node[midway,below] {virtual}; 
    \pdeex{0}{i}{a}\pexci{4}{j}{b}
    
    \draw(.5,1) -- (.5,-1);
    \draw (3.5,1) -- (3.5,-1);
    
    \draw (-.3,0) node[above]{a)};
    \tikzset{shift={(5.5,0)}}
    \draw[line width=2pt, gray] (-.5,0) -- (4.5,0) node[midway,above]{occupied} node[midway,below] {virtual}; 
    \pdeex{0}{i}{a}\pexci{4}{j}{b}
    \draw (.5,1) -- (3.5,1);
    \draw (.5,-1) -- (3.5,-1);
    \draw (-.3,0) node[above]{b)};
    \tikzset{shift={(5.5,0)}}
    \draw[line width=2pt, gray] (-.5,0) -- (4.5,0) node[midway,above]{occupied} node[midway,below] {virtual}; 
    \pdeex{0}{i}{a}\pexci{4}{j}{b}
    \draw (.5,1) .. controls (1,1) and  (3,-1) .. (3.5,-1);
    \draw (.5,-1) .. controls (1,-1) and  (3,1) .. (3.5,1);
    \draw (-.3,0) node[above]{c)};
\end{tikzpicture}
\end{center}
\caption{The only three possible fully connected diagrams corresponding to the $\hat{D}_a^i\hat{E}_j^b$ chain.}
\label{fig:one_connection}
\end{figure}
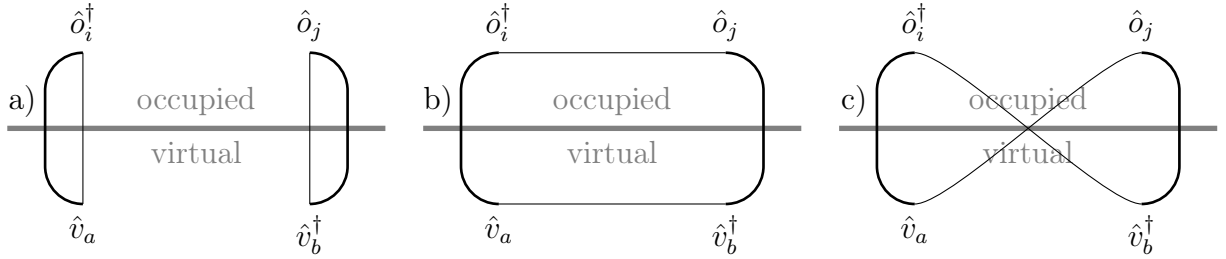
\end{example}

\begin{definition}[Crossing]
In a fully connected diagram, we name {\normalfont crossing} any non-tangential contact between two connection lines joigning extremities of elements of the diagram.
\end{definition}

\subsubsection{Alternative diagrammatic formulation of the corollary of Wick's theorem}
A fully contracted chain containing (de)excitation operators with or without one arbitrary creation-annihilation operator pair can be bijectively associated to a fully connected diagram. Indeed, the diagrams construction rules have been designed such that each connection in the fully connected diagram corresponds to a contraction in the chain of operators.

\begin{proposition}\label{prop:nullity_one_contraction} 
Let $\hat{C}$ be a fully contracted chain containing (de)excitation operators with or without one arbitrary creation-annihilation operator pair. Let $D$ be its associated fully connected diagram. 
The expectation value of $\hat{C}$ relatively to the Fermi vacuum is equal to zero if at least one connection in $D$ does not respect at least one of the following criteria:
\begin{itemize}
\item[P1] Both extremities belong to the same half-plane of the diagram;
\item[P2] The connection is such that the left extremity is pointing to the right and the right extremity is pointing to the left.
\end{itemize}
If $D$ involves a dash symbol, the expectation value of $\hat{C}$ relatively to the Fermi vacuum is equal to zero if the following criterion is not met:
\begin{itemize}
\item[P3] The dash symbol is connected in one of the four ways reproduced in Figure \ref{fig:adaga_graph}.
\end{itemize} 
\end{proposition}

\begin{proof}[\textit{Proof for P3}]
Each of the four patterns in Figure \ref{fig:adaga_graph} corresponds to one of the two contraction schemes described in formulas \eqref{eq:contraction_occ} to \eqref{eq:contraction_in_chain-vir}. Knowing that (i) any other contraction scheme leads to a zero expectation value, and that (ii) any other pattern than those reproduced in Figure \ref{fig:adaga_graph} would correspond to a contraction scheme different from those described in formulas \eqref{eq:contraction_occ} to \eqref{eq:contraction_in_chain-vir}, we deduce that Figure \ref{fig:adaga_graph} presents the only patterns leading to a possibly non-zero expectation value.
\end{proof}

\begin{proof}[Proof for P1]
A connection between two extremities belonging to different half-planes represent a contraction between an \textit{o-}operator and a \textit{v-}operator. According to \eqref{eq:contraction_occ} and \eqref{eq:contraction_vir} the expectation value relatively to the Fermi vacuum of such a contraction is equal to zero.
\end{proof}

\begin{proof}[Proof for P2]
If conditions in \textit{P1} and \textit{P3} are not fulfilled, we already know that the expectation value of $\hat{C}$ is equal to zero. 

Consider that conditions in \textit{P1} and \textit{P3} are fulfilled. In the upper (respectively, lower) half-plane of the diagram, extremities pointing to the right correspond to an \textit{o-}creation (respectively, \textit{v-}annihilation) operator, and extremities pointing to the left correspond to an \textit{o-}annihilation (respectively, \textit{v-}creation) operator. Considering that \textit{P1} and \textit{P3} are fulfilled, we have only three possibilities: (i) the left extremity points to the right and the right extremity points to the left, which corresponds to criterion \textit{P2}; (ii) a connection of two extremities pointing in the same direction corresponds to a contraction of two creation or of two annihilation operators; (iii) the left extremity points to the left and the right extremity points to the right. In cases (ii) and (iii) the corresponding contraction leads to a zero expectation value due to \eqref{eq:contraction_occ} and \eqref{eq:contraction_vir}.
\end{proof}
\begin{definition}
We name {\normalfont well-formed fully connected diagram} a diagram whose structure satisfies property $\textit{P1}$ to $\textit{P3}$ reported in {\normalfont Proposition \ref{prop:nullity_one_contraction}}. 
\end{definition}
\begin{definition}
A chain corresponding to a well-formed fully connected diagram will hereafter be termed {\normalfont well-formed fully contracted chain}.
\end{definition}

\begin{proposition}\label{prop:read_diag_full_contracted}
Let $\hat{C}$ be a well-formed fully contracted chain containing (de)excitation operators with or without one arbitrary creation-annihilation operator pair. Let $D$ be its associated well-formed fully connected diagram. 

A connection in $D$ linking one extremity annotated with an operator pointing at the $\varphi_r$ spin-orbital and an other extremity pointing at the $\varphi_s$ spin-orbital reads $\delta_{r,s}$.

The unsigned expectation value corresponding to $D$ is the product of all the Kronecker deltas corresponding to the connections in $D$.

Let $\mathcal{N}$ be the number of crossings in $D$. The signature of the expectation value corresponding to $D$ is $(-1)^{\mathcal{N}}$.
\end{proposition}

\begin{proof}
Let $\hat{C}_S$ be the chain of second quantization operator in  which, reading it from the right to the left, all the \textit{o}-operators of $\hat{C}$ are first met in the same order as in $\hat{C}$, then all the \textit{v-}operators are then met in the same order as in $\hat{C}$.

If $\hat{C}$ is a chain containing (de)excitation operators with or without one arbitrary creation-annihilation operator pair, an even number of permutations is needed to transform $\hat{C}$ into $\hat{C}_S$. Consequently, a well-formed full contraction of $\hat{C}$ and of $\hat{C}_S$ will share the same signature. 

\noindent In any well-formed full contraction of $\hat{C}_S$, there is no contraction between \textit{o-} and \textit{v-}operators. Accordingly, there is no crossing between any contraction joining two \textit{o-}operators and any contraction joining two \textit{v-}operators. The number of crossings in the well-formed full contraction of $\hat{C}_S$ is the same as the number of crossings in $D$. Let $\mathcal{N}$ be this number. According to what precedes — see paragraph \ref{subsub:corwick} —, the signature is given by $(-1)^{\mathcal{N}}$.

We know that each connection in $D$ corresponds to a contraction in a full contraction of the $\hat{C}$ chain. In agreement with \eqref{eq:contraction_occ} to \eqref{eq:contraction_in_chain-vir}, each connection can be unequivocally interpreted as a Kronecker delta, and $D$ valuates to a product of Kronecker deltas.
\end{proof}

\noindent In order to illustrate what precedes, we study two examples: In Example \ref{ex:full_connection_DE} the chain is composed solely of (de)excitation operators; In Example \ref{ex:full_connection_D_adaga_E} the chain involves a pair of arbitrary creation-annihilation operators.

\begin{example}\label{ex:full_connection_DE}
Let $(i,j,k,l)$ be a 4-tuple of integers, all belonging to $\llbracket 1,N\rrbracket$. 
Let $(a,b,c,d)$ be a 4-tuple of integers, all belonging to $\llbracket N+1,L\rrbracket$. Consider the $\hat{D}_a^i\hat{D}_b^j\hat{E}_k^c\hat{E}_l^d = \hat{o}^\dag_i\hat{v}_a\hat{o}^\dag_j\hat{v}_b\hat{v}^\dag_c\hat{o}_k\hat{v}^\dag_d\hat{o}_l$ chain. The translation of $\hat{D}_a^i\hat{D}_b^j\hat{E}_k^c\hat{E}_l^d$ chain is $\bm{(\,(\;)\, )}$. This translation being a Dyck word the expectation value relatively to the Fermi vacuum is susceptible of being non-zero. In {\normalfont Figure \ref{fig:contraction_2deex-2exci_decomposition}} we provide the $\mathcal{L}_{1,1}$-diagram representation of $\hat{D}_a^i\hat{D}_b^j\hat{E}_k^c\hat{E}_l^d$.

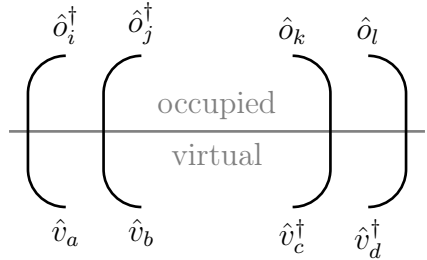
\begin{figure}[h!]
\centering
\begin{tikzpicture}

\draw[color=gray, line width=1pt] (-.25,0) -- (5.25,0) node[midway,above] {occupied} node[midway,below] {virtual};

\pdeex{0}{i}{a}
\pdeex{1}{j}{b}
\pexci{4}{k}{c}
\pexci{5}{l}{d}

\def\b{.5}
\def\r{.5}
\def\x{0}

\end{tikzpicture}
\caption{$\mathcal{L}_{1,1}$-diagram representation of $\hat{D}_a^i\hat{D}_b^j\hat{E}_k^c\hat{E}_l^d$.}
\label{fig:contraction_2deex-2exci_decomposition}
\end{figure}

\noindent There are four well-formed full contractions of the $\hat{o}^\dag_i\hat{v}_a\hat{o}^\dag_j\hat{v}_b\hat{v}^\dag_c\hat{o}_k\hat{v}^\dag_d\hat{o}_l$ chain: 

\begin{align*} 
\braket{\hat{C}_1}_{\Psi_0} &\coloneqq
\braket{\acontraction[1ex]{}{\hat{o}}{^\dag_i\hat{v}_a\hat{o}^\dag_j\hat{v}_b\hat{v}^\dag_c}{\hat{o}}
\acontraction[3ex]{\hat{o}^\dag_i\hat{v}_a}{\hat{o}}{^\dag_j\hat{v}_b\hat{v}^\dag_c\hat{o}_k\hat{v}^\dag_d}{\hat{o}}
\acontraction[2ex]{\hat{o}^\dag_i}{\hat{v}}{_a\hat{o}^\dag_j\hat{v}_b}{\hat{v}}
\acontraction[4ex]{\hat{o}^\dag_i\hat{v}_a\hat{o}^\dag_j}{\hat{v}}{_b\hat{v}^\dag_c\hat{o}_k}{\hat{v}}
\hat{o}^\dag_i\hat{v}_a\hat{o}^\dag_j\hat{v}_b\hat{v}^\dag_c\hat{o}_k\hat{v}^\dag_d\hat{o}_l}_{\Psi_0}
=\delta_{i,k}\delta_{j,l}\delta_{a,c}\delta_{b,d}
,\\%\label{eq:decomposition_fully_contracted_DDEE_1}\\
\braket{\hat{C}_2}_{\Psi_0} &\coloneqq
\braket{\acontraction[1ex]{}{\hat{o}}{^\dag_i\hat{v}_a\hat{o}^\dag_j\hat{v}_b\hat{v}^\dag_c\hat{o}_k\hat{v}^\dag_d}{\hat{o}}
\acontraction[3ex]{\hat{o}^\dag_i\hat{v}_a}{\hat{o}}{^\dag_j\hat{v}_b\hat{v}^\dag_c}{\hat{o}}
\acontraction[2ex]{\hat{o}^\dag_i}{\hat{v}}{_a\hat{o}^\dag_j\hat{v}_b}{\hat{v}}
\acontraction[4ex]{\hat{o}^\dag_i\hat{v}_a\hat{o}^\dag_j}{\hat{v}}{_b\hat{v}^\dag_c\hat{o}_k}{\hat{v}}
\hat{o}^\dag_i\hat{v}_a\hat{o}^\dag_j\hat{v}_b\hat{v}^\dag_c\hat{o}_k\hat{v}^\dag_d\hat{o}_l}_{\Psi_0}
=- \delta_{i,l}\delta_{j,k}\delta_{a,c}\delta_{b,d}
,\\
\braket{\hat{C}_3}_{\Psi_0} &\coloneqq
\braket{\acontraction[1ex]{}{\hat{o}}{^\dag_i\hat{v}_a\hat{o}^\dag_j\hat{v}_b\hat{v}^\dag_c}{\hat{o}}
\acontraction[3ex]{\hat{o}^\dag_i\hat{v}_a}{\hat{o}}{^\dag_j\hat{v}_b\hat{v}^\dag_c\hat{o}_k\hat{v}^\dag_d}{\hat{o}}
\acontraction[4ex]{\hat{o}^\dag_i\hat{v}_a\hat{o}^\dag_j}{\hat{v}}{_b}{\hat{v}}
\acontraction[2ex]{\hat{o}^\dag_i}{\hat{v}}{_a\hat{o}^\dag_j\hat{v}_b\hat{v}^\dag_c\hat{o}_k}{\hat{v}}
\hat{o}^\dag_i\hat{v}_a\hat{o}^\dag_j\hat{v}_b\hat{v}^\dag_c\hat{o}_k\hat{v}^\dag_d\hat{o}_l}_{\Psi_0}
=- \delta_{i,k}\delta_{j,l}\delta_{a,d}\delta_{b,c}
,\\
\braket{\hat{C}_4}_{\Psi_0} &\coloneqq
\braket{\acontraction[1ex]{}{\hat{o}}{^\dag_i\hat{v}_a\hat{o}^\dag_j\hat{v}_b\hat{v}^\dag_c\hat{o}_k\hat{v}^\dag_d}{\hat{o}}
\acontraction[2ex]{\hat{o}^\dag_i}{\hat{v}}{_a\hat{o}^\dag_j\hat{v}_b\hat{v}^\dag_c\hat{o}_k}{\hat{v}}
\acontraction[3ex]{\hat{o}^\dag_i\hat{v}_a}{\hat{o}}{^\dag_j\hat{v}_b\hat{v}^\dag_c}{\hat{o}}
\acontraction[4ex]{\hat{o}^\dag_i\hat{v}_a\hat{o}^\dag_j}{\hat{v}}{_b}{\hat{v}}
\hat{o}^\dag_i\hat{v}_a\hat{o}^\dag_j\hat{v}_b\hat{v}^\dag_c\hat{o}_k\hat{v}^\dag_d\hat{o}_l}_{\Psi_0}
=\delta_{i,l}\delta_{j,k}\delta_{a,d}\delta_{b,c}
.%\label{eq:decomposition_fully_contracted_DDEE_4}
\end{align*}

\noindent According to the corollary of Wick's theorem — see {\normalfont Theorem \ref{theo:Wick}} —, the expectation value of the $\hat{D}_a^i\hat{D}_b^j\hat{E}_k^c\hat{E}_l^d$ chain relatively to the Fermi vacuum can be decomposed as 
\begin{align}
\braket{\hat{D}_a^i\hat{D}_b^j\hat{E}_k^c\hat{E}_l^d}_{\Psi_0}  &= \sum_{i=1}^4\braket{\hat{C}_i}_{\Psi_0}\nonumber\\
&=  \delta_{i,k}\delta_{j,l}\delta_{a,c}\delta_{b,d}
- \delta_{i,l}\delta_{j,k}\delta_{a,c}\delta_{b,d}
- \delta_{i,k}\delta_{j,l}\delta_{a,d}\delta_{b,c}
+ \delta_{i,l}\delta_{j,k}\delta_{a,d}\delta_{b,c}.\label{eq:ddeewick}
\end{align}

\noindent We have just shown that finding all the well-formed full contractions of a chain is equivalent to finding all the well-formed fully connected diagrams. In our case we easily see that, by hand, only four ways exist to connect all extremities, respecting $P1$ and $P2$ introduced in {\normalfont Proposition \ref{prop:nullity_one_contraction}}. They are represented in {\normalfont Figure \ref{fig:fully_contraction_2deex-2exci}}.

\begin{figure}[h!]
\centering
\begin{tikzpicture}
%\draw[grid] (-1,-2) grid (13,2);

\def\b{.5}
\def\r{.5}
\def\x{0}

%top non-cross bottom non-cross&
\draw (-.5,1) node[above]{a)};
\draw[color=gray, line width=1pt] (-.25,0) -- (5.25,0) node[midway,above] {occupied} node[midway,below] {virtual};

\draw[line width=1pt] (\x,\b)   node[left] {$\hat{o}^\dag_i$}  -- (\x,-\b)   node[left] {$\hat{v}_a$};
\draw[line width=1pt] (\x+1,\b) node[left] {$\hat{o}^\dag_j$}  -- (\x+1,-\b) node[left] {$\hat{v}_b$};
\draw[line width=1pt] (\x+4,\b) node[right] {$\hat{o}_k$} -- (\x+4,-\b) node[right] {$\hat{v}^\dag_c$};
\draw[line width=1pt] (\x+5,\b) node[right] {$\hat{o}_l$} -- (\x+5,-\b) node[right] {$\hat{v}^\dag_d$};

\draw[line width=1pt] (\x,\b)    arc (180:0:2 and 1);
\draw[line width=1pt] (\x+1,\b)  arc (180:0:2 and 1);
\draw[line width=1pt] (\x,-\b)   arc (180:360:2 and 1);
\draw[line width=1pt] (\x+1,-\b) arc (180:360:2 and 1);

\tikzset{shift={(7,0)}}

%top non-cross bottom cross
\draw (-.5,1) node[above]{b)};
\draw[color=gray, line width=1pt] (-.25,0) -- (5.25,0) node[midway,above] {occupied} node[midway,below] {virtual};

\draw[line width=1pt] (\x,\b)   node[left] {$\hat{o}^\dag_i$}  -- (\x,-\b)   node[left] {$\hat{v}_a$};
\draw[line width=1pt] (\x+1,\b) node[left] {$\hat{o}^\dag_j$}  -- (\x+1,-\b) node[left] {$\hat{v}_b$};
\draw[line width=1pt] (\x+4,\b) node[right] {$\hat{o}_k$} -- (\x+4,-\b) node[right] {$\hat{v}^\dag_c$};
\draw[line width=1pt] (\x+5,\b) node[right] {$\hat{o}_l$} -- (\x+5,-\b) node[right] {$\hat{v}^\dag_d$};

\draw[line width=1pt] (\x,\b)    arc (180:0:2.5 and 1.5);
\draw[line width=1pt] (\x+1,\b)  arc (180:0:1.5 and 1);
\draw[line width=1pt] (\x,-\b)   arc (180:360:2 and 1);
\draw[line width=1pt] (\x+1,-\b) arc (180:360:2 and 1);

\tikzset{shift={(-7,-4)}}

%top cross bottom non-cross
\draw (-.5,1) node[above]{c)};
\draw[color=gray, line width=1pt] (-.25,0) -- (5.25,0) node[midway,above] {occupied} node[midway,below] {virtual};

\draw[line width=1pt] (\x,\b)   node[left] {$\hat{o}^\dag_i$}  -- (\x,-\b)   node[left] {$\hat{v}_a$};
\draw[line width=1pt] (\x+1,\b) node[left] {$\hat{o}^\dag_j$}  -- (\x+1,-\b) node[left] {$\hat{v}_b$};
\draw[line width=1pt] (\x+4,\b) node[right] {$\hat{o}_k$} -- (\x+4,-\b) node[right] {$\hat{v}^\dag_c$};
\draw[line width=1pt] (\x+5,\b) node[right] {$\hat{o}_l$} -- (\x+5,-\b) node[right] {$\hat{v}^\dag_d$};

\draw[line width=1pt] (\x,\b)    arc (180:0:2 and 1);
\draw[line width=1pt] (\x+1,\b)  arc (180:0:2 and 1);
\draw[line width=1pt] (\x,-\b)   arc (180:360:2.5 and 1.5);
\draw[line width=1pt] (\x+1,-\b) arc (180:360:1.5 and 1);

\tikzset{shift={(7,0)}}

%top non-cross bottom non-cross&\\
\draw (-.5,1) node[above]{d)};
\draw[color=gray, line width=1pt] (-.25,0) -- (5.25,0) node[midway,above] {occupied} node[midway,below] {virtual};

\draw[line width=1pt] (\x,\b)   node[left] {$\hat{o}^\dag_i$}  -- (\x,-\b)   node[left] {$\hat{v}_a$};
\draw[line width=1pt] (\x+1,\b) node[left] {$\hat{o}^\dag_j$}  -- (\x+1,-\b) node[left] {$\hat{v}_b$};
\draw[line width=1pt] (\x+4,\b) node[right] {$\hat{o}_k$} -- (\x+4,-\b) node[right] {$\hat{v}^\dag_c$};
\draw[line width=1pt] (\x+5,\b) node[right] {$\hat{o}_l$} -- (\x+5,-\b) node[right] {$\hat{v}^\dag_d$};

\draw[line width=1pt] (\x,\b)    arc (180:0:2.5 and 1.5);
\draw[line width=1pt] (\x+1,\b)  arc (180:0:1.5 and 1);
\draw[line width=1pt] (\x,-\b)   arc (180:360:2.5 and 1.5);
\draw[line width=1pt] (\x+1,-\b) arc (180:360:1.5 and 1);
\end{tikzpicture}

\caption{The four possible well-formed fully connected diagrams corresponding to the $\hat{D}_a^i\hat{D}_b^j\hat{E}_k^c\hat{E}_l^d$ chain.}
\label{fig:fully_contraction_2deex-2exci}
\end{figure}
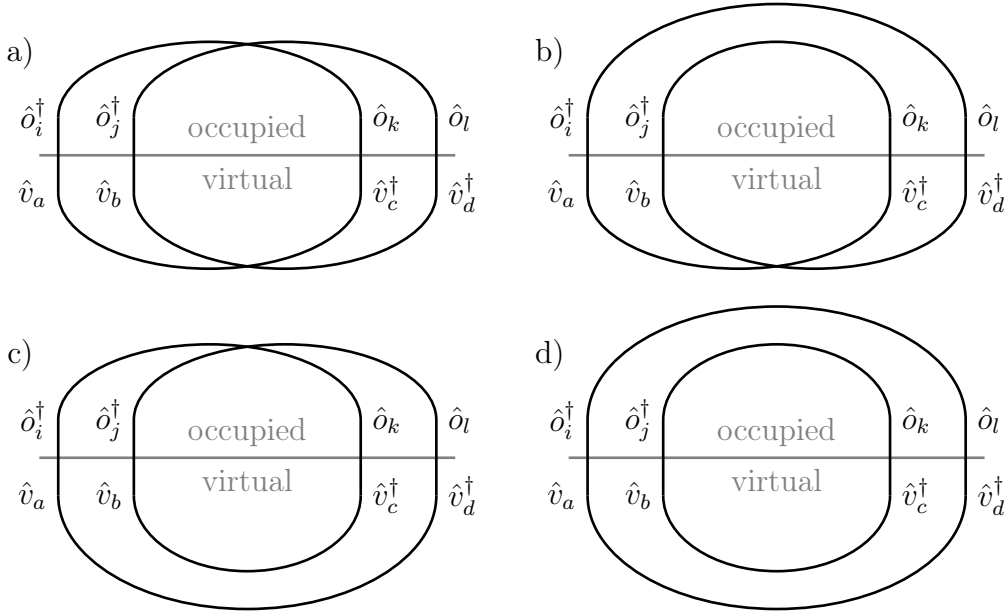
\end{example}

\noindent According to {\normalfont Proposition \ref{prop:read_diag_full_contracted}}, each well-formed fully connected diagram corresponds to a product of Kronecker deltas. The diagrams represented in parts a) to d) of {\normalfont Figure \ref{fig:fully_contraction_2deex-2exci}} correspond to the expectation values $\braket{\hat{C}_1}_{\Psi_0}$ to $\braket{\hat{C}_4}_{\Psi_0}$, respectively. Diagrams a) and d) have an even number of crossings — two and zero, respectively — so that the corresponding signature is ``+''. On the other hand, diagrams b) and c) have one crossing, so that the corresponding signature is ``$-$''. In conclusion, the exact same result derived using {\normalfont Theorem \ref{theo:Wick}} in \eqref{eq:ddeewick} can be found from the $\mathcal{L}_{1,1}$-diagram representation.

\begin{proposition}\label{prop:symmetry}
Let $\hat{C}$ be a chain solely composed of (de)excitation operators. Let $D$ be the $\mathcal{L}_{1,1}$-diagram representation of $\hat{C}$. Consider a connection respecting $P1$ and $P2$ in {\normalfont Proposition \ref{prop:nullity_one_contraction}} in one half-plane of $D$. The image of that connection in the other half-plane using the imaginary line as symmetry axis also respects $P1$ and $P2$.
\end{proposition}

\begin{proof}
The two extremities of the connection are in a single half-plane before and after the symmetry operation is performed. Moreover, $P2$ satisfaction is preserved due to the fact that the left-right directionality is preserved by the axial symmetry operation described in the proposition.
\end{proof}

\noindent We can build a simplified systematic method for finding all the well-formed fully connected diagrams corresponding to any chain, say $\hat{C}$, of (de)excitation operators:

\begin{ourrule}\label{rule:symmetry}
Consider the upper half-plane of the $\mathcal{L}_{1,1}$-diagram representation of $\hat{C}$, and connect all the extremities using connections that respect the $\textit{P2}$ criterion introduced in Proposition \ref{prop:nullity_one_contraction}. This step produces $K$ connection patterns in the upper half-plane. According to {\normalfont Proposition \ref{prop:symmetry}} one can find $K$ connection patterns for the lower half-plane. One finally combines one pattern for each half-plane, leading to $K^2$ well-formed fully connected diagrams.
\end{ourrule}

\begin{remark}
From {\normalfont Proposition \ref{prop:symmetry}} and {\normalfont Rule \ref{rule:symmetry}} one sees that it is necessary only to find $K$ connection patterns for systematically deducing the $K^2$ well-formed fully connected diagrams corresponding to the chain of interest. This $K$ number will be given an explicit expression in the next section of this paper.
\end{remark}

\begin{example}\label{ex:full_connection_D_adaga_E}
Let $(i,j)$ be a couple of integers, both belonging to $\llbracket 1,N\rrbracket$. 
Let $(a,b)$ be a couple of integers, both belonging to $\llbracket N+1,L\rrbracket$.
Let $\hat{a}^\dag_r$ and $\hat{a}_s$ be two second quantization operators, each pointing at a spin-orbital with non-definite occupation number relatively to the Fermi vacuum.
Consider the $\hat{D}_a^i\hat{a}^\dag_r\hat{a}_s\hat{E}_j^b = \hat{o}^\dag_i\hat{v}_a\hat{a}^\dag_r\hat{a}_s\hat{v}^\dag_b\hat{o}_j$ chain. Its translation is $\bm{(\,-\, )}$. According to {\normalfont Proposition IV.5} in {\normalfont Part I}, the expectation value of $\hat{D}_a^i\hat{a}^\dag_r\hat{a}_s\hat{E}_j^b$ relatively to the Fermi vacuum can be different from zero. In {\normalfont Figure \ref{fig:deex-adaga-exci}} we already provided the $\mathcal{L}_{1,1}$-diagram representation of $\hat{D}_a^i\hat{a}^\dag_r\hat{a}_s\hat{E}_j^b$.

\noindent There are three well-formed full contractions of the $\hat{o}^\dag_i\hat{v}_a\hat{a}^\dag_r\hat{a}_s\hat{v}^\dag_b\hat{o}_j$ chain: 

\begin{align*}\label{eq:contraction_deex-adaga-exci_decomposition}
\braket{\hat{C}_1}_{\Psi_0}
&\coloneqq \braket{
\acontraction[1ex]{}{\hat{o}}{^\dag_i\hat{v}_a\hat{a}^\dag_r\hat{a}_s\hat{v}^\dag_b}{\hat{o}}
\acontraction[2ex]{\hat{o}^\dag_i}{\hat{v}}{_a}{\hat{a}}
\acontraction[2ex]{\hat{o}^\dag_i\hat{v}_a\hat{a}^\dag_r}{\hat{a}}{_s}{\hat{v}}
\hat{o}^\dag_i\hat{v}_a\hat{a}^\dag_r\hat{a}_s\hat{v}^\dag_b\hat{o}_j
}_{\Psi_0}
= \delta_{a,s}\delta_{b,r}\delta_{i,j}(1-n_r^{\Psi_0})(1-n_s^{\Psi_0}),
\quad\\
\braket{\hat{C}_2}_{\Psi_0}
&\coloneqq \braket{
\acontraction[1ex]{}{\hat{o}}{^\dag_i\hat{v}_a\hat{a}^\dag_r}{\hat{a}}
\acontraction[2ex]{\hat{o}^\dag_i\hat{v}_a}{\hat{a}}{^\dag_r\hat{a}_s\hat{v}^\dag_b}{\hat{o}}
\acontraction[3ex]{\hat{o}^\dag_i}{\hat{v}}{_a\hat{a}^\dag_r\hat{a}_s}{\hat{v}}
\hat{o}^\dag_i\hat{v}_a\hat{a}^\dag_r\hat{a}_s\hat{v}^\dag_b\hat{o}_j
}_{\Psi_0}
=-\delta_{i,r}\delta_{j,s}\delta_{a,d}\delta_{b,c}n_r^{\Psi_0}n_s^{\Psi_0},
\quad\\
\braket{\hat{C}_3}_{\Psi_0}
&\coloneqq \braket{
\acontraction[1ex]{}{\hat{o}}{^\dag_i\hat{v}_a\hat{a}^\dag_r\hat{a}_s\hat{v}^\dag_b}{\hat{o}}
\acontraction[2ex]{\hat{o}^\dag_i}{\hat{v}}{_a\hat{a}^\dag_r\hat{a}_s}{\hat{v}}
\acontraction[3ex]{\hat{o}^\dag_i\hat{v}_a}{\hat{a}}{^\dag_r}{\hat{a}}
\hat{o}^\dag_i\hat{v}_a\hat{a}^\dag_r\hat{a}_s\hat{v}^\dag_b\hat{o}_j
}_{\Psi_0}
=\delta_{i,j}\delta_{a,b}\delta_{r,s}n_r^{\Psi_0}n_s^{\Psi_0}.
\end{align*}

\noindent According to the corollary of Wick's theorem — see {\normalfont Theorem \ref{theo:Wick}} —, the expectation value of the $\hat{D}_a^i\hat{a}^\dag_r\hat{a}_s\hat{E}_j^b$ chain relatively to the Fermi vacuum can be decomposed as

\begin{align*}
\braket{\hat{D}_a^i\hat{a}^\dag_r\hat{a}_s\hat{E}_j^b}_{\Psi_0} &= \sum_{i=1}^3 \braket{\hat{C}_i}_{\Psi_0} \\
&= \delta_{a,r}\delta_{b,s}\delta_{i,j}(1-n^{\Psi_0}_r)(1-n^{\Psi_0}_s)
- \delta_{i,s}\delta_{j,r}\delta_{a,b}n^{\Psi_0}_rn^{\Psi_0}_s
+ \delta_{i,j}\delta_{a,b}\delta_{r,s}n^{\Psi_0}_rn^{\Psi_0}_s
\end{align*}

\noindent Observing {\normalfont Figure \ref{fig:deex-adaga-exci}} one sees that there are only three ways to connect all extremities, respecting $P1$, $P2$ and $P3$ introduced in {\normalfont Proposition \ref{prop:nullity_one_contraction}}. They are represented in {\normalfont Figure \ref{fig:contraction_deex-adaga-exci_decomposition}}.

\begin{figure}[h!] 
\centering
 \resizebox {\textwidth} {!} {
\begin{tikzpicture}
%\draw[grid] (-1,2) grid (6,-16);

%%%
\def\b{.5}
\def\r{.5}
\def\x{0}

\draw[color=white] (-.25,0) -- (5.25,0);

\draw[color=white,line width=1pt] (\x,\b)
	-- (\x,-\b) arc (180:360:2.5 and 1.5)
	-- (\x+5,\b) arc (0:180:2.5 and 1.5);
\draw[color=white,line width=1pt] (\x+3,0) arc (-90:0:\r) 
	-- ++(0,0) arc (0:180:1 and .5)
	-- ++(0,0) arc (180:270:\r) 
	-- cycle;
\draw[color=white] (2.5,-2.3) node {$+\delta_{i,j}\delta_{a,b}\delta_{r,s}n_r^{\Psi_0}n_s^{\Psi_0}$};

\draw (-.5,1) node[above]{a)};
\draw[color=gray, line width=1pt] (-.55,0) -- (5.55,0) node[midway,above] {occupied} node[midway,below] {virtual};
%\draw (-.5,1.25) node {b)};

\draw[line width=1pt] (\x-.3,\b) node[left] {$\hat{o}^\dag_i$}
	-- (\x-.3,-\b) node[left] {$\hat{v}_a$} arc (180:270:\r) 
	-- (\x+.2,-1) arc (-90:0:2*\r) node[above] {$\hat{a}^\dag_r$}
	-- (\x+3.8,0) node[above] {$\hat{a}_s$} arc (180:270:2*\r)
	-- (\x+4.8,-1) arc (-90:0:\r)  node[right] {$\hat{v}^\dag_b$}
	-- (\x+5.3,\b) node[right] {$\hat{o}_j$} arc (0:180:2.8 and 1.5);

%%%
\tikzset{shift={(7,0)}}

\draw (-.5,1) node[above]{b)};
\draw[color=gray, line width=1pt] (-.25,0) -- (5.25,0) node[midway,above] {occupied} node[midway,below] {virtual};
%\draw (-.5,1.25) node {c)};

\draw[line width=1pt] (\x,\b) node[left] {$\hat{o}^\dag_i$}
	-- (\x,-\b) node[left] {$\hat{v}_a$} arc (180:360:2.5 and 1.5) node[right] {$\hat{v}^\dag_b$}
	-- (\x+5,\b) node[right] {$\hat{o}_j$} arc (0:180:1.75 and 1.5)
	-- ++(0,0) node[left] {$\hat{a}^\dag_r$} arc (180:270:\r) 
	-- (\x+3,0) arc (-90:0:\r) node[right] {$\hat{a}_s$}
	-- ++(0,0) arc (0:180:1.75 and 1.5);

%%%
\tikzset{shift={(7,0)}}

\draw (-.5,1) node[above]{c)};
\draw[color=gray, line width=1pt] (-.25,0) -- (5.25,0) node[midway,above] {occupied} node[midway,below] {virtual};
%\draw (-.5,1.25) node {d)};

\draw[line width=1pt] (\x,\b) node[left] {$\hat{o}^\dag_i$}
	-- (\x,-\b) node[left] {$\hat{v}_a$} arc (180:360:2.5 and 1.5) node[right] {$\hat{v}^\dag_b$}
	-- (\x+5,\b) node[right] {$\hat{o}_j$} arc (0:180:2.5 and 1.5);
\draw[line width=1pt] (\x+3,0) arc (-90:0:\r) node[right] {$\hat{a}_s$}
	-- ++(0,0) arc (0:180:1 and .5) node[left] {$\hat{a}^\dag_r$}
	-- ++(0,0) arc (180:270:\r) 
	-- cycle;
\end{tikzpicture}
}
\vspace*{-1cm}
\caption{The three well-formed fully connected diagrams corresponding to the $\hat{D}_a^i\hat{a}^\dag_r\hat{a}_s\hat{E}_j^b$ chain.}
\label{fig:contraction_deex-adaga-exci_decomposition}
\end{figure}
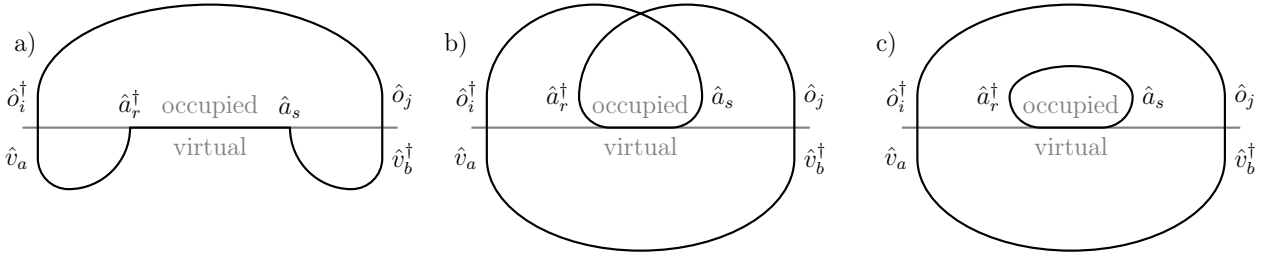

\noindent According to {\normalfont Proposition \ref{prop:read_diag_full_contracted}}, all the well-formed fully connected diagrams represented in the {\normalfont Figure \ref{fig:contraction_deex-adaga-exci_decomposition}} correspond to a product of Kronecker deltas. The diagram a), b) and c) correspond to the $\braket{\hat{C}_1}_{\Psi_0}$, $\braket{\hat{C}_2}_{\Psi_0}$ and $\braket{\hat{C}_3}_{\Psi_0}$ expectation values, respectively. Therefore, we find again the same result applying {\normalfont Theorem \ref{theo:Wick}} and our alternative diagrammatic procedure.
\end{example}

\section{Numbering}
\subsection{Chains of second quantization operators}

As we have seen in Part I, there are some sufficient conditions for the nullity of the expectation value of a chain of second quantization which do not require the identification of which spinorbital each operator is pointing at.
There exists a mapping between the sequence of second quantization operators in a chain which possibly has a non-zero expectation value relatively to a reference state, and a Dyck word.
\begin{definition}[Catalan number]
The $n^\text{th}$ Catalan number $C_n$ is defined as
\begin{equation*}
C_n = \dfrac{1}{n+1}\binom{2n}{n}.
\end{equation*}
\end{definition}

\begin{proposition}\label{prop:catalan-dyck}
For a given value of a natural integer $n$, there are exactly $C_n$ Dyck words that can be composed with $n$ opening brackets and $n$ closing brackets.
\end{proposition}
\begin{proof}
A proof can be found in \cite{roman_introduction_2015}.
\end{proof}
\noindent The three following propositions provide the number of sequences of second quantization operators in chains which possibly have a non-zero expectation value relatively to the physical (respectively, Fermi) vacuum — Proposition \ref{prop:numbering_physic} — (respectively, propositions \ref{prop:numbering_fermi} and \ref{prop:numbering_fermi_ex_dex}).

\begin{proposition}\label{prop:numbering_physic}
There are $C_n(n!)^2$ ways to arrange $n$ creation and $n$ annihilation operators in order to get a $2n$--long chain with a possibly non-zero expectation value relatively to the physical vacuum. 
\end{proposition}
\begin{proof}
Let $(r_k,s_k)_{1\leq k \leq n}$ be a $2n$-tuple of integers, all belonging to $\llbracket 1, L\rrbracket$. Consider the $n$ annihilation operators $\hat{a}_{s_1},\ldots,\hat{a}_{s_n}$, and the $n$ creation operators $\hat{a}^\dag_{r_1},\ldots,\hat{a}^\dag_{r_n}$.

We know from Part I that if the $\mathcal{L}_1$-translation of a chain of $n$ creation and $n$ annihilation operators is not a Dyck word its expectation value relatively to the physical vacuum is zero. 

Consider a Dyck word, $S$, composed of $n$ opening and $n$ closing brackets, and the unique $\hat{C}$ chain satisfying the two following conditions: (i) for every $m$ in $\llbracket 1, 2n\rrbracket$, the $m^\text{th}$ operator in $\hat{C}$ is an annihilation (respectively, creation) operator if the $m^\text{th}$ bracket in $S$ is an opening (respectively, closing) bracket; (ii) for every $k$ in $\llbracket 1, n\rrbracket$, the $k^\text{th}$ annihilation (respectively, creation) operator in $\hat{C}$ is $\hat{a}_{s_k}$ (respectively, $\hat{a}^\dag_{r_k}$). The $\mathcal{L}_1$--translation of $\hat{C}$ is $S$.

Moreover, all the chains obtained by permuting creation operators or by permuting annihilation operators — here we use an inclusive ``or'' — in $\hat{C}$ have in common that their $\mathcal{L}_1$--translation is $S$. Consequently, the number of ways to arrange $n$ creation and $n$ annihilation operators in a chain whose expectation value relatively to the physical vacuum is possibly non-zero is equal to the number of Dyck words composed with $n$ opening and $n$ closing brackets, i.e., $C_n$ — see Proposition \ref{prop:catalan-dyck} —, times the number of ways to permute all creation operators — they are $n!$ —, times the number of ways to permute all annihilation operators, i.e., $(n!)$.
\end{proof}

\begin{example}
Let $(r_1,s_1,r_2,s_2)$ be a 4-tuple of integers, all belonging to $\llbracket 1,L\rrbracket$. Consider the two annihilation operators $\hat{a}_{r_1}$ and $\hat{a}_{r_2}$, and the two creation operators $\hat{a}^\dag_{s_1}$ and $\hat{a}^\dag_{s_2}$. 

{\normalfont Proposition \ref{prop:catalan-dyck}} tells us that there are two Dyck words composed of $2$ opening and $2$ closing brackets.
The first one is $\bm{((\,))}$; the associated chain as defined in the proof of {\normalfont Proposition \ref{prop:numbering_physic}} is $\hat{a}_{r_1}\hat{a}_{r_2}\hat{a}^\dag_{s_1}\hat{a}^\dag_{s_2}$. Considering all permutations obtained from this chain by permuting creation operators or by permuting annihilation operators, the four chains whose $\mathcal{L}_1$--translation is $\bm{((\,))}$ are
\begin{equation}
\hat{a}_{r_1}\hat{a}_{r_2}\hat{a}^\dag_{s_1}\hat{a}^\dag_{s_2}, \;
\hat{a}_{r_1}\hat{a}_{r_2}\hat{a}^\dag_{s_2}\hat{a}^\dag_{s_1}, \;
\hat{a}_{r_2}\hat{a}_{r_1}\hat{a}^\dag_{s_1}\hat{a}^\dag_{s_2}, \;
\hat{a}_{r_2}\hat{a}_{r_1}\hat{a}^\dag_{s_2}\hat{a}^\dag_{s_1}.
\end{equation}
The second one is $\bm{(\,)(\,)}$; the associated chain as defined in the proof of {\normalfont Proposition \ref{prop:numbering_physic}} is $\hat{a}_{r_1}\hat{a}^\dag_{s_1}\hat{a}_{r_2}\hat{a}^\dag_{s_2}$. Considering all permutations obtained from this chain by permuting creation operators or by permuting annihilation operators, the four chains whose $\mathcal{L}_1$--translation is $\bm{(\,)(\,)}$ are
\begin{equation}
\hat{a}_{r_1}\hat{a}^\dag_{s_1}\hat{a}_{r_2}\hat{a}^\dag_{s_2}, \;
\hat{a}_{r_1}\hat{a}^\dag_{s_2}\hat{a}_{r_2}\hat{a}^\dag_{s_1}, \;
\hat{a}_{r_2}\hat{a}^\dag_{s_1}\hat{a}_{r_1}\hat{a}^\dag_{s_2}, \;
\hat{a}_{r_2}\hat{a}^\dag_{s_2}\hat{a}_{r_1}\hat{a}^\dag_{s_1}.
\end{equation}
In conclusion, we indeed find the number of chains with a possibly non-zero expectation value — they are eight — we expected to find using the formula in {\normalfont Proposition \ref{prop:numbering_physic}}.
\end{example}

\begin{remark}
In a more abstract approach in which the indices to the operators are not variables but parameters, i.e., an approach in which we only keep the nature (``creation'' vs. ``annihilation'') of the operators, Catalan numbers count the number of sequences of creation and annihilation operators in a chain — i.e., the number of patterns — leading to an expectation value that can be non-zero.
\end{remark}

\begin{proposition}\label{prop:numbering_fermi}
The number of ways to arrange $n_o$ \textit{o-}creation operators, $n_o$ \textit{o-}annihilation operators, $n_v$ \textit{v-}creation operators, and $n_v$ \textit{v-}annihilation operators into a $(2n_o+2n_v)$--long chain whose expectation value relatively to the Fermi vacuum is possibly different from zero is equal to
\begin{equation*}
C_{n_o}C_{n_v}\binom{2n_o+2n_v}{2n_o}(n_o!)^2(n_v!)^2.
\end{equation*}
\end{proposition}
\begin{proof}
Let $\hat{C}$ a chain composed of $n_o$ \textit{o-}creation, $n_o$ \textit{o-}annihilation, $n_v$ \textit{v-}creation, and $n_v$ \textit{v-}annihilation operators. Let $S_o\&S_v$ be the split-$\mathcal{L}_2$ translation of $\hat{C}$.
We know from Part I that if $S_o\&S_v$ is not a pair of Dyck words the expectation value $\braket{\hat{C}}_{\Psi_0}$ is equal to zero. 
Following the proof of {\normalfont Proposition \ref{prop:numbering_physic}}, the \textit{o-} (respectively, \textit{v-})operators can be arranged in $C_{n_o}(n_o!)^2$ (respectively, $C_{n_v}(n_v!)^2)$ ways in order for $S_o$ (respectively, $S_v$) to be a Dyck word.
Among the $2n_o+2n_v$ operators in the $\hat{C}$ chain, $2n_o$ are \textit{o-}operators, the $2n_v$ remaining are \textit{v-}operators.
For a given sequence of \textit{o-}operators and \textit{v-}operators there are $\binom{2n_o+2n_v}{2n_o}$ ways to arrange them in order to form a $\hat{C}$ chain with $\braket{\hat{C}}_{\Psi_0}$ possibly being different from zero. 
\end{proof}

\begin{proposition}\label{prop:numbering_fermi_ex_dex}
There are $C_n(n!)^2$ ways to arrange $n$ excitation and $n$ deexcitation operators in order to get a $2n$--long chain with a possibly non-zero expectation value relatively to the Fermi vacuum. 
\end{proposition}
\begin{proof}
If a chain of $n$ excitation and $n$ deexcitation operators does not Fermi-$\mathcal{L}_1$ translate into a Dyck word its expectation value relatively to the Fermi vacuum is zero. This result can be obtained by following the proof of {\normalfont Proposition \ref{prop:numbering_physic}}.
\end{proof}

\begin{example}\label{ex:numbering}
Let $(i,j)$ be a couple of integers, both belonging to $\llbracket 1,N\rrbracket$. Let $(a,b)$ be a couple of integers, both belonging to $\llbracket N+1,L\rrbracket$. Consider the four following operators: $\hat{o}_i^\dag$, $\hat{o}_j$, $\hat{v}_a$, and $\hat{v}_b^\dag$. According to {\normalfont Proposition \ref{prop:numbering_fermi}} there are six ways to arrange these four operators in a chain with a possibly non-zero expectation value relatively to the Fermi vacuum. Indeed, the six chains are
\begin{equation*}
\hat{o}_i^\dag\hat{o}_j\hat{v}_a\hat{v}_b^\dag, \quad
\hat{o}_i^\dag\hat{v}_a\hat{o}_j\hat{v}_b^\dag, \quad
\hat{v}_a\hat{o}_i^\dag\hat{o}_j\hat{v}_b^\dag, \quad
\hat{v}_a\hat{v}_b^\dag\hat{o}_i^\dag\hat{o}_j, \quad
\hat{v}_a\hat{o}_i^\dag\hat{v}_b^\dag\hat{o}_j, \quad
\hat{o}_i^\dag\hat{v}_a\hat{v}_b^\dag\hat{o}_j,
\end{equation*}
their expectation value relatively to the Fermi vacuum is equal to $\pm\delta_{i,j}\delta_{a,b}$. The four individual operators can be used to define exactly one deexcitation operator $\hat{D}_a^i=\hat{o}_i^\dag\hat{v}_a$ and exactly one excitation operator $\hat{E}_j^b=\hat{v}_b^\dag\hat{o}_j$. According to {\normalfont Proposition \ref{prop:numbering_fermi_ex_dex}} there exists only one way to arrange these two operators in order to obtain a chain with a possibly non-zero expectation value relatively to the Fermi vacuum: $\hat{D}_a^i\hat{E}_j^b$.
\end{example}

\subsection{Corollary of Wick's theorem}
The corollary of Wick's theorem — see Theorem \ref{theo:Wick} — states that any expectation value of a chain of second quantization operators can be decomposed into a sum of fully contracted chains expectation values. 

\begin{definition}
Let $\hat{C}$ be a chain of excitation and deexcitation operators. Then, $\mathcal{W}(\hat{C})$ denotes the number of possible ``well'' full contractions of $\hat{C}$, i.e., a full contraction about which we cannot say, using solely \eqref{eq:contraction_occ} to \eqref{eq:contraction_in_chain-vir}, that its expectation value is equal to zero.
\end{definition}

\begin{remark}
From what has been learned in {\normalfont Part  I} we readily deduce that if a $\hat{C}$ chain's Fermi-$\mathcal{L}_1$ translation is not a Dyck word, then $\mathcal{W}(\hat{C})$ is equal to zero.
\end{remark}

\begin{proposition}
Let $\hat{C}$ be a chain composed of $n$ excitation operators and $n$ deexcitation operators whose Fermi-$\mathcal{L}_1$ translation, $S$, is a Dyck word. Let $(d_k^\uparrow)_{1\leq k\leq n}$ be the list of opening depths of $S$. Then, the $\mathcal{W}(\hat{C})$ number is equal to the square of the product of opening depths, i.e.,
\begin{equation}
\mathcal{W}(\hat{C}) = \left( \prod_{k=1}^n d_k^\uparrow\right)^{\!2}.
\end{equation}
\end{proposition}

\begin{proof}
From formulas \eqref{eq:contraction_occ} to \eqref{eq:contraction_in_chain-vir}, we know that in order to get a $\braket{\hat{C}}_{\Psi_0}$ value different from zero we have to contract each \textit{o-}creation (respectively \textit{v-}annihilation) with an \textit{o-}annihilation (respectively \textit{v-}creation) on its right. 
Let $S_o\& S_v$ be the $\mathcal{L}_2$ translation of $\hat{C}$. The $\hat{C}$ chain being composed solely of excitation and deexcitation operators, both $S_o$ and $S_v$ have the same sequence of opening and closing brackets as $S$. Consequently, their lists of opening depths of $S$, $S_o$, and $S_v$ are the same. Contracting two operators in order to get a possible non-zero expectation value is equivalent to expulsing the corresponding opening and closing brackets in $S_o$ or $S_v$. 
A full contraction of the $\hat{C}$ chain corresponds to one opening-closing bracket pairing scheme in $S_o$ and $S_v$ in view of fully expulsing $S_o$ and $S_v$. According to Proposition III.5  in Part I, the number of ways to fully expulse $S_o$ is equal to the number of ways to fully expulse $S_v$, which is equal to the product of opening depths of $S$.
 Therefore the number of ways to ``well'' fully contract the $\hat{C}$ chain of operators is equal to the square of the product of opening depths of $S$.
\end{proof}

\begin{remark}\label{rem:N_not_number_of_electrons}
In the following proposition, $N$ is a variable number that should not be confused with the number of electrons.
\end{remark}

\begin{proposition}\label{prop:numbering_wick}
Let $\hat{C}$ be a chain composed of $n$ excitation operators and $n$ deexcitation operators. 
Consider a new chain, $\hat{C}^N_{K}$, constructed by introducing $N$ arbitrary creation operators followed by $N$ arbitrary annihilation operators after — i.e., on the right of — the $K^\text{th}$ operator in $\hat{C}$ read from the left to the right. Then,

$\;$

\noindent{\normalfont (i)} If the Fermi-$\mathcal{L}_1$ translation of $\hat{C}$ is not a Dyck word, then $\mathcal{W}(\hat{C}^N_{K})$ is equal to zero. 

\noindent{\normalfont (ii)} If the Fermi-$\mathcal{L}_1$ translation of $\hat{C}$ is a Dyck word, then $\mathcal{W}(\hat{C}^N_{K})$ is equal to
\begin{equation}
N!\mathcal{M}^N_{d_K}\mathcal{W}(\hat{C}),
\end{equation}
where $\mathcal{M}_{d_K}^N$ is a multiplicative factor that depends on both the number of arbitrary operators and their insertion depth:
\begin{equation}\label{eq:multiplicateur}
\mathcal{M}_{d_K}^N = \sum_{\ell=0}^{min(N,{d_K})}\binom{N}{\ell}^2\binom{{d_K}+N-\ell}{N},
\end{equation}
where $d_K$ is the depth reached after the $K^\mathrm{th}$ position in the Fermi-$\mathcal{L}_1$ translation of $\hat{C}$.
\end{proposition}

\begin{proof}
The combinatorial details of the proof are provided in Appendix.
\end{proof}

\begin{example}
Let $(i,j)$ be a couple of integers, both belonging to $\llbracket 1,N\rrbracket$. 
Let $(a,b)$ be a couple of integers, both belonging to $\llbracket N+1,L\rrbracket$.
Let $\hat{a}^\dag_{r}$, and $\hat{a}_{s}$ be two second quantization operators, each pointing at a spin-orbital with non-definite occupation number relatively to the Fermi vacuum. 

Consider the $\hat{D}_{a}^{i}\hat{E}_{j}^{b}$ chain, hereafter named $\hat{C}$. Its translation is $\bm{(\;)}$. The corresponding list of depths is $(0,1,0)$, and the corresponding list of opening depth $(1)$.

We introduce the $\hat{a}^\dag_{r}\hat{a}_{s}$ pair of operators after the $\hat{D}_{a}^{i}$ deexcitation operator. The depth reached after this operator is equal to one, the resulting chain is named $\hat{C}_1^1 = \hat{D}_{a}^{i}\hat{a}^\dag_{r}\hat{a}_{s}\hat{E}_{j}^{b}$.

According to {\normalfont Proposition \ref{prop:numbering_wick}}, $\mathcal{W}(\hat{C}_1^1)$ is equal to three. Indeed, there exist three well-formed full contractions of $\hat{C}^1_1$. They were already represented in {\normalfont Example \ref{ex:full_connection_D_adaga_E}}.
\end{example}

\begin{remark}
For $N$ going from $1$ to $10$, formula {\normalfont(\ref{eq:multiplicateur})} corresponds to the crystal ball sequence for the $A_N$ root lattice {\normalfont\cite{conway_lowdimensional_1997}}. For small values of $N$, simplified expressions for $\mathcal{M}^N_d$ can be derived:
\begin{equation*}
\mathcal{M}^1_d = 2d+1,		\quad      % = (d+1)^2 - d^2,\quad
\mathcal{M}^2_d = 3d(d+1)+1,	\quad% = (d+1)^3 - d^3,\quad
%\mathcal{M}^3_d = \dfrac{(2d+1)(5d^2+5d+3)}{3},\quad
\mathcal{M}^3_d = \dfrac{d(10d^2+15d+11)}{3}+1.\quad
%\mathcal{M}^4_d = \dfrac{5d(d+1)(7d^2+7d+10)}{12}+1.
\end{equation*}
\end{remark}

\section{Commutator simplification using a diagrammatic representation}

In the presentation done in this section we limit ourselves to (nested) commutators involving one creation-annihilation operator pair per slot, with at least one pair being an excitation or a deexcitation operator.

\subsection{Diagrammatic representation of a commutator}

\begin{definition}[Diagrammatic representation of a commutator]\label{def:commutator_representation}
Let $\hat{C}_1$ and $\hat{C}_2$ be two pairs of creation-annihilation operators. The {\normalfont diagrammatic representation} of the $[\hat{C}_1,\hat{C}_2]$ commutator  is the $\mathcal{L}_{1,1}$ diagrammatic representation of $\hat{C}_1$ and $\hat{C}_2$, in this order, separated by the {\normalfont $\copyright$} symbol and placed between two square brackets.
The diagrammatic representation of $\hat{C}_1$ is named the {\normalfont left part} of the diagram and the diagrammatic representation of $\hat{C}_2$ is named the {\normalfont right part} of the diagram.
\end{definition}

\begin{example}
Let $(i,j)$ be a couple of integers, both belonging to $\llbracket 1,N\rrbracket$. 
Let $(a,b)$ be a couple of integers, both belonging to $\llbracket N+1,L\rrbracket$. 
The diagrammatic representation of the $[\hat{D}_a^i,\hat{E}_j^b]$ commutator is depicted in {\normalfont Figure \ref{fig:commu_E-D}}.
\begin{figure}[h!]
\begin{center}
\begin{tikzpicture}[scale=0.7]
\draw (-.5,1.75) -- ++(-.25,0) -- ++(0,-3.5) -- ++(.25,0);
\pdeex{0}{i}{a}\pexci{3}{j}{b} 
\draw (1.5,0) node {$\copyright$};
\draw (3.5,1.75) -- ++(.25,0) -- ++(0,-3.5) -- ++(-.25,0);
\end{tikzpicture}
\end{center}
\caption{Diagrammatic representation of the $[\hat{D}_a^i,\hat{E}_j^b]$ commutator.}
\label{fig:commu_E-D}
\end{figure}
\end{example}

\begin{definition}[Diagrammatic representation of nested commutator]\label{def:nested_commutator_representation}
The {\normalfont diagrammatic representation of the commutator} is obtained by successively constructing the diagrammatic representation from the innermost commutator to the outermost commutator following {\normalfont Definition \ref{def:commutator_representation}}.
\end{definition}

\begin{example}
Let $(i,j)$ be a couple of integers, both belonging to $\llbracket 1,N\rrbracket$. 
Let $(a,b)$ be a couple of integers, both belonging to $\llbracket N+1,L\rrbracket$. 
Let $\hat{a}^\dag_r$ and $\hat{a}_s$ be two second quantization operators, each pointing at a spin-orbital with non-definite occupation number relatively to the Fermi vacuum.
The diagrammatic representation of the two nested commutators $[ [ \hat{a}^\dag_r\hat{a}_s, \hat{E}_i^a ], \hat{E}_j^b ]$ and $[ [ \hat{o}^\dag_i\hat{v}_a, \hat{a}^\dag_s\hat{a}_r ], \hat{v}^\dag_b\hat{o}_j]$ are illustrated in Figure \ref{fig:double_commu}.
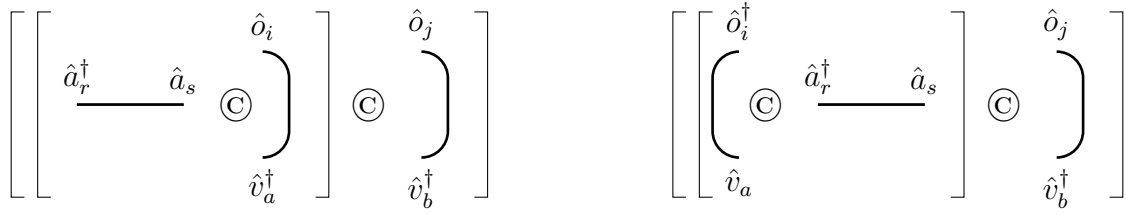
\begin{figure}[h!]
\begin{center}
\begin{tikzpicture}[scale=0.7]
\draw (-2,1.75) -- ++(-.25,0) -- ++(0,-3.5) -- ++(.25,0);
\draw (-1.5,1.75) -- ++(-.25,0) -- ++(0,-3.5) -- ++(.25,0);
\draw (2,0) node {$\copyright$};
\padaga{0}{}{r}{}{s}\pexci{3}{i}{a} 
\draw (3.5,1.75) -- ++(.25,0) -- ++(0,-3.5) -- ++(-.25,0);
\tikzset{shift={(3,0)}}
\draw (1.5,0) node {$\copyright$};
\tikzset{shift={(2,0)}}
\pexci{1}{j}{b}
\draw (1.5,1.75) -- ++(.25,0) -- ++(0,-3.5) -- ++(-.25,0);

\tikzset{shift={(6,0)}}
\draw (0,1.75) -- ++(-.25,0) -- ++(0,-3.5) -- ++(.25,0);
\draw (-.5,1.75) -- ++(-.25,0) -- ++(0,-3.5) -- ++(.25,0);
\draw (1,0) node {$\copyright$};
\pdeex{0}{i}{a}\padaga{3}{}{r}{}{s} 
\draw (4.5,1.75) -- ++(.25,0) -- ++(0,-3.5) -- ++(-.25,0);
\tikzset{shift={(4,0)}}
\draw (1.5,0) node {$\copyright$};
\tikzset{shift={(2,0)}}
\pexci{1}{j}{b}
\draw (1.5,1.75) -- ++(.25,0) -- ++(0,-3.5) -- ++(-.25,0);
\end{tikzpicture}
\end{center}
\caption{Diagrammatic representation of $[ [ \hat{a}^\dag_r\hat{a}_s, \hat{E}_i^a ], \hat{E}_j^b ]$ (left), and of $[ [ \hat{o}^\dag_i\hat{v}_a, \hat{a}^\dag_s\hat{a}_r ], \hat{v}^\dag_b\hat{o}_j]$ (right).}
\label{fig:double_commu}
\end{figure}
\end{example}

\begin{remark}
As we see comparing figures \ref{fig:commu_E-D} and \ref{fig:double_commu}, whilst for simple commutators the ``{\normalfont $\copyright\!$}'' and ``{\normalfont $[\;]\!$}'' notations are strictly redundant, this is not the case for nested commutators.
However, the outermost square brackets in a diagram can always be ignored without loss of information.
\end{remark}

\newpage\subsection{Diagrammatic simplification of a simple commutator}

A commutator involving a pair of creation-annihilation operators and a (de)excitation operator can be simplified through a diagrammatic representation.

\begin{ourrule}[Conditioning of a commutator involving at least one (de)excitation operator]\label{rule:preconditioning}
Consider $[\hat{C}_1, \hat{C}_2]$ a commutator involving at least one (de)excitation operator. We have three possible cases: \\

Case i. If $\hat{C_1}$ is the only excitation operator and/or if $\hat{C_2}$ is the only deexcitation operator the {\normalfont conditioned commutator} associated with $[\hat{C}_1, \hat{C}_2]$ is $[\hat{C}_2, \hat{C}_1]$.\\

Case ii. If $\hat{C_2}$ is the only excitation operator and/or if $\hat{C_1}$ is the only deexcitation operator the {\normalfont conditioned commutator} associated with $[\hat{C}_1, \hat{C}_2]$ is $[\hat{C}_1, \hat{C}_2]$.\\

Case iii. If $\hat{C}_1$ and $\hat{C}_2$ are both excitation operators, or if $\hat{C}_1$ and $\hat{C}_2$ are both deexcitation operators, the {\normalfont conditioned commutator} associated with $[\hat{C}_1, \hat{C}_2]$ is $[\hat{C}_1, \hat{C}_2]$.
\end{ourrule}

\noindent A commutator involving at least an excitation or deexcitation operator can be diagrammatically developed using its diagrammatic representation.

\begin{ourrule}[Diagrammatic development of a conditioned commutator involving at least one (de)excitation operator]\label{rule:develop_commu}
The diagrammatic representation of the conditioned commutator involving at least one (de)excitation operator is constructed following {\normalfont Definition \ref{def:commutator_representation}}. All diagrams obtained by linking with a single connection the left and right part of the diagrammatic representation of the commutator are built. Each connection must respect conditions \textit{P1}, \textit{P2}, and \textit{P3} reported in {\normalfont Proposition \ref{prop:nullity_one_contraction}}. 

If a connection involves the upper extremity of the diagrammatic representation of an excitation (respectively, deexcitation) operator, the corresponding bracket is pulled to the left (respectively, right) in order to put the lower extremity on the most left position of the diagram. Diagrams obtained from this procedure are separated by the ``$\bigoplus$'' symbol.
\end{ourrule}

\noindent The three following examples illustrate the application of Rule \ref{rule:develop_commu}.

\begin{example}\label{ex:commu_A-E}
Let $(i,a)$ be a couple of integers belonging to $\llbracket 1,N\rrbracket \times \llbracket N+1,L\rrbracket$.
Let $\hat{a}^\dag_r$ and $\hat{a}_s$ be two second quantization operators, each pointing at a spin-orbital with non-definite occupation number relatively to the Fermi vacuum. Consider the $[\hat{a}^\dag_r\hat{a}_s,\hat{E}_i^a]$ commutator.
It involves an excitation operator and a pair of arbitrary creation-annihilation operators. The excitation operator being already in the right position in the commutator, the conditioned commutator corresponding to $[\hat{a}^\dag_r\hat{a}_s,\hat{E}_i^a]$ is $[\hat{a}^\dag_r\hat{a}_s,\hat{E}_i^a]$ itself.
After building the diagrammatic representation of the commutator, two possible connections exist, both depicted in {\normalfont Figure \ref{fig:commu_A-E_dev} a)}. The first connection implies extremities annotated with operators $\hat{a}_s$ and $\hat{v}_a^\dag$. The second one implies extremities annotated with operators $\hat{a}_r^\dag$ and $\hat{o}_i$. This connection involves the upper extremity of the bracket corresponding to the excitation operators. This bracket has to be pulled on the left of the diagram, in order to set the lower extremity on the most left position in the diagram.
The diagrammatic development of the $[\hat{a}^\dag_r\hat{a}_s,\hat{E}_i^a]$ commutator is illustrated in {\normalfont Figure \ref{fig:commu_A-E_dev} b)}.
\begin{figure}[h!]
\centering

\begin{tikzpicture}[scale=0.7]
\draw (-2,0) node {a)};
\padaga{0}{}{r}{}{s}\tcr{\pexci{3}{i}{a}}
\draw (2,0) node {$\copyright$};

\tikzset{shift={(2,2.5)}}

\draw[->,thick] (1.5,-1) -- (3,0) --node[above,black] {Connection of} node[below] {$\hat{a}_s$ and $\hat{v}_a^\dag$} ++(4,0);
\tikzset{shift={(7.5,0)}}
\adaga{none}{r}{deex}{s}\tcr{\exciend{i}{a}}

\tikzset{shift={(-10,-5)}}

\draw[->,thick] (1.5,1) -- (3,0) --node[above,black] {Connection of} node[below] {$\hat{a}_r^\dag$ and $\hat{o}_i$} ++(4,0);
\tikzset{shift={(9.5,0)}}
\draw[line width=1pt] (0,0) node[above] { $\hat{a}_s$} 
	-- (-1.5,0) node[above] { $\hat{a}_r^\dag$} arc(270:90:.5)
	-- (.3,1); 
\draw[line width=1pt,red] (.3,1)node[above] { $\hat{o}_i$} arc(90:0:.5)
	-- (.8,-.5) arc(0:-90:.5)node[below] { $\hat{v}^\dag_a$}; 
\draw[->,thick] (1.7,0) --node[above] {pulling} ++(2.5,0);
\tikzset{shift={(4.7,0)}}
\tcr{\exci{i}{a}}\tikzset{shift={(1,0)}}\adaga{deex}{r}{}{s}
\end{tikzpicture}

\begin{tikzpicture}[scale=0.7]
\draw (-2,0) node {b)};
\padaga{0}{}{r}{}{s}\pexci{3}{i}{a} 
\draw (2,0) node {$\copyright$};
\draw[->,thick] (4,0) --node[above] {Development} ++(3,0);
\tikzset{shift={(8,0)}}
\adaga{none}{r}{deex}{s}\exciend{i}{a} 
\draw (1.5,0) node {$\bigoplus$};
\tikzset{shift={(2,0)}}
\exci{i}{a}\adaga{deex}{r}{a}{s}
\end{tikzpicture}
\caption{a) The possible connections between the two parts of the diagrammatic representation of the commutator as described in {\normalfont Rule \ref{rule:develop_commu}}; b) Diagrammatic development of the $[\hat{a}^\dag_r\hat{a}_s,\hat{E}_i^a]$ commutator.}
\label{fig:commu_A-E_dev}
\end{figure}
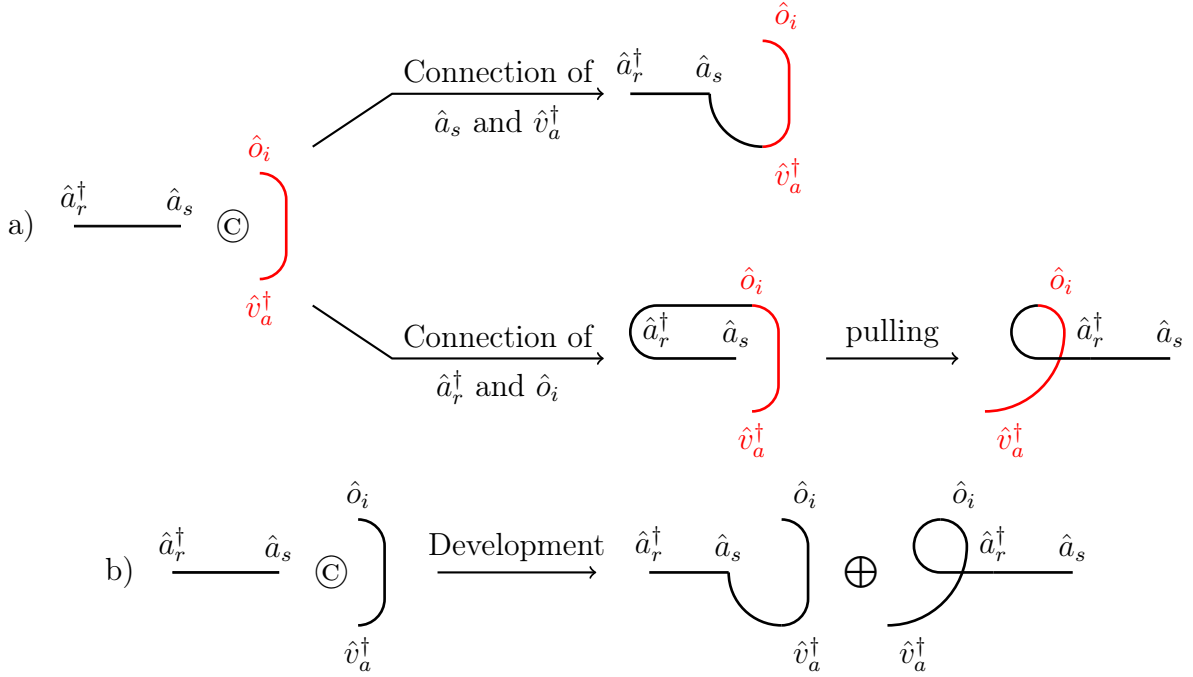
\end{example}

\begin{example}\label{ex:commu_D-A}
Let $(i,a)$ be a couple of integers belonging to $\llbracket 1,N\rrbracket \times \llbracket N+1,L\rrbracket$.
Let $\hat{a}^\dag_r$ and $\hat{a}_s$ be two second quantization operators, each pointing at a spin-orbital with non-definite occupation number relatively to the Fermi vacuum. Consider the $[\hat{a}^\dag_r\hat{a}_s,\hat{D}_a^i]$ commutator.
It involves a deexcitation operator and a pair of arbitrary creation-annihilation operators. The deexcitation operator is in the right position in the commutator, hence the conditioned commutator associated with $[\hat{a}^\dag_r\hat{a}_s,\hat{D}_a^i]$ is
\begin{equation*}
 \left[\hat{D}_a^i,\hat{a}^\dag_r\hat{a}_s\right].
\end{equation*}
After building the diagrammatic representation of the $[\hat{D}_a^i,\hat{a}^\dag_r\hat{a}_s]$ commutator, two possible connections are possible, both depicted in {\normalfont Figure \ref{fig:commu_D-A_dev} a)}. The first connection implies extremities annotated with operators $\hat{v}_a$ and $\hat{a}_r^\dag$. The second one implies extremities annotated with operators $\hat{o}_i^\dag$ and $\hat{a}_s$. This connection involves the upper extremities of the bracket corresponding to the deexcitation operators. This bracket has to be pulled on the right of the diagram, in order to set the lower extremity on the most right position in the diagram.
The diagrammatic development of the $[\hat{D}_a^i, \hat{a}^\dag_r\hat{a}_s]$ commutator is illustrated in {\normalfont Figure \ref{fig:commu_D-A_dev} b)}.
\begin{figure}[h!]
\centering

\begin{tikzpicture}[scale=0.7]
\draw (-2,0) node {a)};
\tcr{\pdeex{-1}{i}{a}} \padaga{2}{}{r}{}{s}
\draw (0,0) node {$\copyright$};

\tikzset{shift={(2,2.5)}}

\draw[->,thick,white] (3,0) --node[above,black] {Connection of} ++(4,0);
\draw[->,thick] (1.5,-1) -- (3,0) --node[below] {$\hat{v}_a$ and $\hat{a}_r^\dag$} ++(4,0);
\tikzset{shift={(8,0)}}
\tcr{\deexend{a}{i}}\adaga{exci}{r}{}{s}

\tikzset{shift={(-10.5,-5)}}

\draw[->,thick,white] (3,0) --node[above,black] {Connection of} ++(4,0);
\draw[->,thick] (1.5,1) -- (3,0) --node[below] {$\hat{o}_i^\dag$ and $\hat{a}_s$} ++(4,0);
\tikzset{shift={(8.5,0)}}
\draw[line width=1pt] (0,0) node[above] { $\hat{a}^\dag_r$} 
	-- (1.5,0) node[above] {$\hat{a}_s$} arc(-90:90:.5)
	-- (-.3,1); 
\draw[line width=1pt, red] (-.3,1)node[above] { $\hat{o}_i^\dag$} arc(90:180:.5)
	-- (-.8,-.5) arc(180:270:.5)node[below] { $\hat{v}_a$}; 
\draw[->,thick] (2.5,0) --node[above] {pulling} ++(2.5,0);
\tikzset{shift={(5.5,0)}}
\adaga{}{r}{exci}{s}\tcr{\deex{a}{i}}
\end{tikzpicture}
%\tikzset{shift={(-17,-5)}}
\begin{tikzpicture}[scale=0.7]
\draw (-2,0) node {b)};
\pdeex{-1}{i}{a} \padaga{2}{}{r}{}{s}
\draw (0,0) node {$\copyright$};
\draw[->,thick] (4,0) --node[above] {Development} ++(2.5,0);
\tikzset{shift={(8.5,0)}}
\deexend{a}{i}\adaga{exci}{r}{none}{s}
\draw (1,0) node {$\bigoplus$};
\tikzset{shift={(2,0)}}
\adaga{none}{r}{exci}{s}\deex{a}{i}
\end{tikzpicture}

\caption{a) The possible connections between the two parts of the diagrammatic representation of the commutator as described in {\normalfont Rule \ref{rule:develop_commu}}; b) Diagrammatic development of the $[\hat{D}_a^i, \hat{a}^\dag_r\hat{a}_s]$ commutator.}
\label{fig:commu_D-A_dev}
\end{figure}
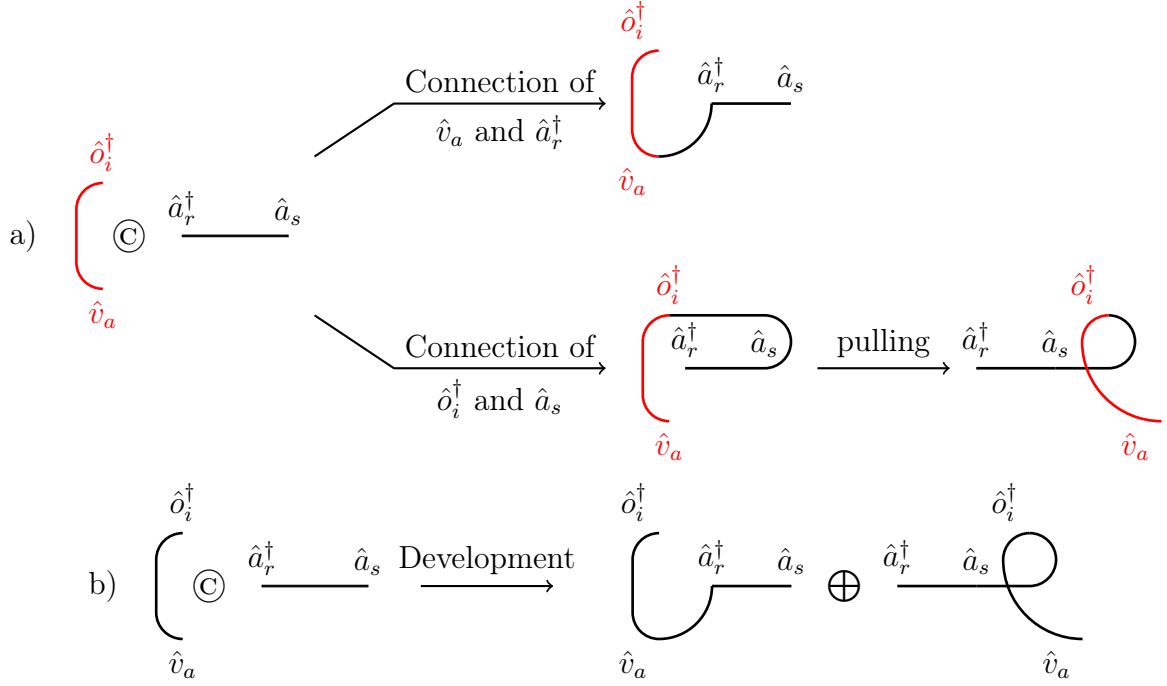
\end{example}

\begin{remark}
We can notice that the two diagrams obtained after the development of the commutator represented in {\normalfont Figure \ref{fig:commu_D-A_dev}} are mirror images of the two diagrams obtained from the development of the commutator involving an excitation operator reported in {\normalfont Figure \ref{fig:commu_A-E_dev} b)}.
\end{remark}

\begin{example}\label{ex:commu_E-E}
Let $(i,j)$ be a couple of integers, both belonging to $\llbracket 1,N\rrbracket$. 
Let $(a,b)$ be a couple of integers, both belonging to $\llbracket N+1,L\rrbracket$.
Consider the $[\hat{E}_i^a,\hat{E}_j^b]$ commutator.
It involves two excitation operators. There is already an excitation operator on the right position of the commutator deexcitation operator, the {\normalfont Rule \ref{rule:develop_commu}} can directly be applied. Looking at {\normalfont Figure \ref{fig:commu_E-E_dev}}, we see that no connection can be built between the two parts of the diagram. The graphical development of the commutator is the empty diagram.

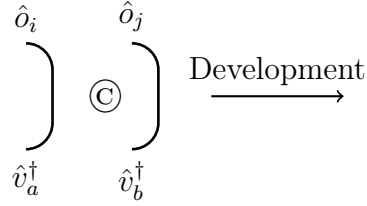
\begin{figure}[h!]
\centering
\begin{tikzpicture}[scale=0.7]
\pexci{0}{i}{a}\pexci{2}{j}{b} 
\draw (1,0) node {$\copyright$};
\draw[->,thick] (3,0) --node[above] {Development} ++(2.5,0);
\end{tikzpicture}

\caption{The diagrammatic representation of the $[\hat{E}_i^a,\hat{E}_j^b]$ commutator is the empty diagram.}
\label{fig:commu_E-E_dev}
\end{figure}
\end{example}

\begin{remark}
{\normalfont Example \ref{ex:commu_E-E}} illustrates the fact that two excitation operators always commute. This observation also applies to the commutator of two deexcitation operators.
\end{remark}

\noindent Definition \ref{def:commutator_representation} and Rule \ref{rule:develop_commu} have been designed in order to ``read'' each diagram as the development of the commutator.

\begin{ourrule}[Interpretation of the diagrammatic development of a conditioned commutator]\label{rule:interpretation_diagram_commutator}
Let $D$ be the diagrammatic development of a conditioned commutator. The $D$ diagram is a juxtaposition of ``sub-diagrams''.

A connection in a sub-diagram linking one extremity annotated with an operator pointing at the $\varphi_p$ spin-orbital and another extremity pointing at the $\varphi_q$ spin-orbital reads $\delta_{p,q}$.

Let the most left extremity in a sub-diagram be annotated with the $\hat{p}^\dag_l$ operator, and the most right extremity be annotated with the $\hat{p}_r$ operator. The unsigned operator corresponding to that sub-diagram is the product of a number — the product of all the Kronecker deltas corresponding to the connections in the sub-diagram —, and the $\hat{p}^\dag_l\hat{p}_r$ pair.

Let $\mathcal{N}$ be the number of crossings in a sub-diagram. The signature of the operator corresponding to that sub-diagram is $(-1)^{\mathcal{N}}$.

The developed conditioned commutator is equal to the summation of signed operators corresponding to every sub-diagram of $D$.
\end{ourrule}

\begin{ourrule}[Simplification of a commutator involving at least one (de)excitation operator]\label{rule:signed_pre-computed_commutator}
The simplification of a commutator involving at least one (de)excitation operator is obtained by interpreting the diagrammatic development of the corresponding conditioned commutator — see {\normalfont Rule \ref{rule:interpretation_diagram_commutator} } —, and by multiplying the result by $(-1)$ if the conditioned commutator has been obtained following case i in {\normalfont Rule \ref{rule:preconditioning}}.
\end{ourrule}

\begin{example}
Consider the three commutators introduced in examples {\normalfont \ref{ex:commu_A-E}, \ref{ex:commu_D-A}}, and {\normalfont \ref{ex:commu_E-E}} and their graphical development represented in figures {\normalfont \ref{fig:commu_A-E_dev}, \ref{fig:commu_D-A_dev}}, and {\normalfont \ref{fig:commu_E-E_dev}}. 
According to {\normalfont Rule \ref{rule:interpretation_diagram_commutator}} and {\normalfont Rule \ref{rule:signed_pre-computed_commutator}}, the three commutators can be simplified as follows:
\begin{align*}
\left[\hat{a}^\dag_r\hat{a}_s,\hat{E}_i^a\right] 
&= \delta_{s,a}\hat{a}_r^\dag\hat{o}_i - \delta_{i,r}\hat{v}_a^\dag\hat{a}_s,\\
\left[\hat{a}^\dag_r\hat{a}_s,\hat{D}_a^i\right] 
&= -\left[\hat{D}_a^i,\hat{a}^\dag_r\hat{a}_s\right]
 = -\delta_{a,r}\hat{o}_i^\dag\hat{a}_s +\delta_{s,i}\hat{a}_r^\dag\hat{v}_a,\\
\left[\hat{E}_i^a,\hat{E}_j^b\right]
&= 0\hat{\mathds{I}}.
\end{align*}
The first commutator can be directly simplified by interpreting its diagrammatic development. On the other hand, the conditioned commutator in {\normalfont Example {\normalfont \ref{ex:commu_D-A}}} has been obtained according to case i. in {\normalfont Rule \ref{rule:preconditioning}}, so the sign of each term obtained when interpreting the diagrammatic development of the conditioned commutator corresponding to $[\hat{a}^\dag_r\hat{a}_s,\hat{D}_a^i]$ has to be changed for simplifying $[\hat{a}^\dag_r\hat{a}_s,\hat{D}_a^i]$. Indeed, the diagrammatic development has been done for $[\hat{D}_a^i,\hat{a}^\dag_r\hat{a}_s]$ and not for $[\hat{a}^\dag_r\hat{a}_s,\hat{D}_a^i]$, both being equal up to a sign. The last commutator, i.e., $[\hat{E}_i^a,\hat{E}_j^b]$ is diagrammatically developed into an empty diagram. As a consequence, the corresponding operator is the zero operator.
\end{example}

\subsection{Diagrammatic simplification of nested commutators}

\begin{ourrule}[Conditioning of nested commutators involving (de)excitation operators]\label{rule:preconditioning_nested}
The {\normalfont conditioned nested commutator} associated with a nested commutator involving (de)excitation operators is obtained by successively conditioning the constitutive commutators from the innermost to the outermost commutator. 
\end{ourrule}

\begin{ourrule}[Diagrammatic development of conditioned nested commutators involving some (de)excitation operators]\label{rule:nest_commu_simpli}
In order to develop conditioned nested commutators involving in the innermost commutator a pair of creation-annihilation operators, and in all other positions (de)excitation operators, {\normalfont  Rule \ref{rule:develop_commu}} is repeated from the innermost commutator to the outermost one. The developed diagram obtained from an inner commutator is used to build the diagrammatic representation of the next outer commutator.
\end{ourrule}

\begin{ourrule}[Interpretation of the diagrammatic development of conditioned nested commutators involving (de)excitation operators]\label{rule:interpretation_diagram_nested_commutator}
The expanded operator corresponding to a conditioned nested commutator involving (de)excitation operators is ``read'' following {\normalfont Rule \ref{rule:interpretation_diagram_commutator}}.
\end{ourrule}

\begin{ourrule}[Simplification of nested commutators involving (de)excitation operators]\label{rule:signed_pre-computed_nested_commutator}
The simplification of nested commutators involving (de)excitation operators is obtained by interpreting the diagrammatic development of the corresponding conditioned nested commutator — see {\normalfont Rule \ref{rule:interpretation_diagram_nested_commutator}} —, and by multiplying the result by $(-1)$ if conditioned commutators have been obtained following case i in {\normalfont Rule \ref{rule:preconditioning}} an odd number of times.
\end{ourrule}

\noindent The two examples given in {\normalfont Figure \ref{fig:double_commu}} will be used to illustrate the diagrammatic simplification of nested commutators involving (de)excitation operators. 

\begin{example}\label{ex:commu_exci_exci}
Let $(i,j)$ be a couple of integers, both belonging to $\llbracket 1,N\rrbracket$. 
Let $(a,b)$ be a couple of integers, both belonging to $\llbracket N+1,L\rrbracket$. 
Let $\hat{a}^\dag_r$ and $\hat{a}_s$ be two second quantization operators, each pointing at a spin-orbital with non-definite occupation number relatively to the Fermi vacuum. Consider the following nested commutators:
\begin{equation}\label{eq:commu_A-E-E}
\left[ \left[ \hat{a}^\dag_r\hat{a}_s, \hat{E}_i^a \right], \hat{E}_j^b \right]. 
\end{equation}

\noindent The two commutators can be developed into four terms, each involving six second quantization operators.

\begin{align*}
\left[ \left[ \hat{a}^\dag_r\hat{a}_s, \hat{E}_i^a \right], \hat{E}_j^b \right] = \left[ \left[ \hat{a}^\dag_r\hat{a}_s, \hat{v}^\dag_a\hat{o}_i \right], \hat{v}^\dag_b\hat{o}_i \right] 
&= 
   \hat{a}^\dag_r\hat{a}_s\hat{v}^\dag_a\hat{o}_i\hat{v}^\dag_b\hat{o}_j 
-  \hat{v}^\dag_a\hat{o}_i\hat{a}^\dag_r\hat{a}_s\hat{v}^\dag_b\hat{o}_j \\
&- \hat{v}^\dag_b\hat{o}_j\hat{a}^\dag_r\hat{a}_s\hat{v}^\dag_a\hat{o}_i
+  \hat{v}^\dag_b\hat{o}_j\hat{v}^\dag_a\hat{o}_i\hat{a}^\dag_r\hat{a}_s. 
\end{align*}
\noindent By applying anticommutation rules, the four terms can be reduced to two terms, each involving two second quantization operator and two Kronecker deltas. 

\begin{equation*}\label{eq:double_commutator_adaga_E_E}
\left[ \left[ \hat{a}^\dag_r\hat{a}_s, \hat{E}_i^a \right], \hat{E}_j^b \right] 
=
-\delta_{j,r}\delta_{a,s}\hat{E}_i^b 
-\delta_{i,r}\delta_{b,s}\hat{E}_j^a .
\end{equation*}

\noindent We will now use the {\normalfont Rule \ref{rule:nest_commu_simpli}} on the diagrammatic representation of the nested commutator -- represented in {\normalfont Figure \ref{fig:double_commu}} -- to diagrammatically derive this result. The inner commutator, $[ \hat{a}^\dag_r\hat{a}_s, \hat{E}_i^a ]$, has already been developed in {\normalfont Example \ref{ex:commu_A-E}}. The diagram development obtained in {\normalfont Figure \ref{fig:commu_A-E_dev}} can be used to build the diagrammatic representation of the outer commutator. It is illustrated in {\normalfont Figure \ref{fig:commu_A-E-simpli-E}}.

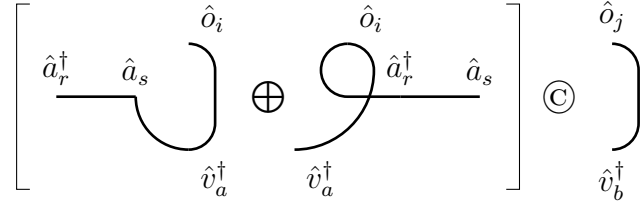
\begin{figure}[h!]
\begin{center}
\begin{tikzpicture}[scale=0.7]
\draw (-.5,1.75) -- ++(-.25,0) -- ++(0,-3.5) -- ++(.25,0);
\adaga{none}{r}{deex}{s}\exciend{i}{a} 
\draw (1.5,0) node {$\bigoplus$};
\tikzset{shift={(2,0)}}
\exci{i}{a}\adaga{deex}{r}{a}{s}
\draw (.5,1.75) -- ++(.25,0) -- ++(0,-3.5) -- ++(-.25,0);
\draw (1.5,0) node {$\copyright$};
\tikzset{shift={(2,0)}}
\pexci{1}{j}{b} 
\end{tikzpicture}
\end{center}
\caption{Graphical representation of nested commutators $[[\hat{a}^\dag_r\hat{a}_s,\hat{E}_i^a ], \hat{E}_j^b]$ with the inner commutator beeing already developed.}
\label{fig:commu_A-E-simpli-E}
\end{figure}
\noindent Each sub-diagram obtained from the inner commutator can form one allowed connection with the right part of the outer commutator diagrammatic representation. They are both represented in {\normalfont Figure \ref{fig:commu_A-E-E_connect}}. 

\begin{figure}[h!]
\begin{center}
\begin{tikzpicture}[scale=0.7]
\draw (-1,0) node {a)};
\adaga{}{r}{deex}{s}\exciend{i}{a}\tcr{\pexci{2.5}{j}{b}}
\draw (1.5,0) node {$\copyright$};
\draw[->,thick,white] (3.5,0) --node[above,black] {Connection of} ++(2.5,0);
\draw[->,thick] (3.5,0) --node[below] {$\hat{a}^\dag_r$ and $\hat{o}_j$} ++(2.5,0);
\tikzset{shift={(6.5,0)}}
\tcr{\exci{j}{b}}\tikzset{shift={(1,0)}}\adaga{deex}{r}{deex}{s}\exciend{i}{a}
\end{tikzpicture}
%\tikzset{shift={(-14,-5)}}
\begin{tikzpicture}[scale=0.7]
\draw (-.5,0) node {b)};
\exci{i}{a}\adaga{deex}{r}{}{s}\tcr{\pexci{2}{j}{b}}
\draw (1,0) node {$\copyright$};
\draw[->,thick,white] (3,0) --node[above,black] {Connection of} ++(2.5,0);
\draw[->,thick] (3,0) --node[below] {$\hat{a}_s$ and $\hat{v}_b^\dag$} ++(2.5,0);
\tikzset{shift={(6,0)}}
\exci{i}{a}\adaga{deex}{r}{deex}{s}\tcr{\exciend{j}{b}}
\end{tikzpicture}
\end{center}
\caption{a) and b) represent the only two allowed connections between the two sub-diagrams obtained from the inner commutator diagrammatic development, and the $\mathcal{L}_{1,1}$ diagrammatic representation of the $\hat{E}_j^b$ operator.}
\label{fig:commu_A-E-E_connect}
\end{figure}
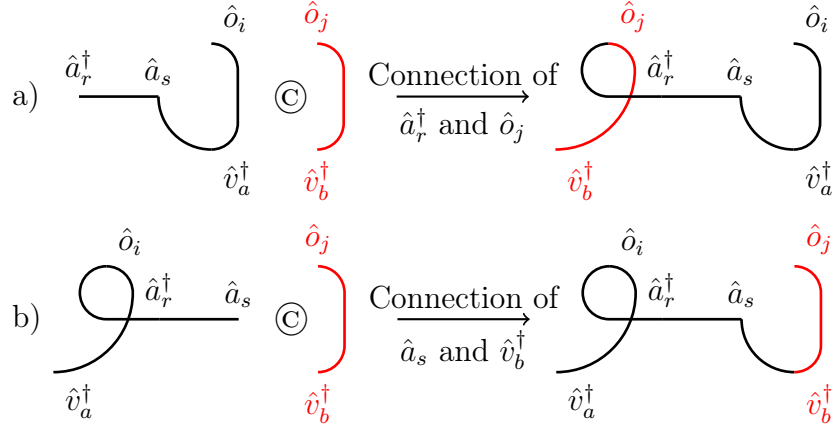
The diagrammatic development of the nested commutator defined in \eqref{eq:commu_A-E-E} is represented in {\normalfont Figure \ref{fig:commu_A-E-E_dev}}. {\normalfont Figure \ref{fig:commu_A-E-E_dev}} can be read following {\normalfont Rule \ref{rule:interpretation_diagram_commutator}}:

\begin{equation*}
\left[ \left[ \hat{a}^\dag_r\hat{a}_s, \hat{E}_i^a \right], \hat{E}_j^b \right] 
=
-\delta_{j,r}\delta_{a,s}\hat{v}_b^\dag\hat{o}_i
-\delta_{i,r}\delta_{b,s}\hat{v}_a^\dag\hat{o}_j
=
-\delta_{j,r}\delta_{a,s}\hat{E}_i^b 
-\delta_{i,r}\delta_{b,s}\hat{E}_j^a.
\end{equation*}
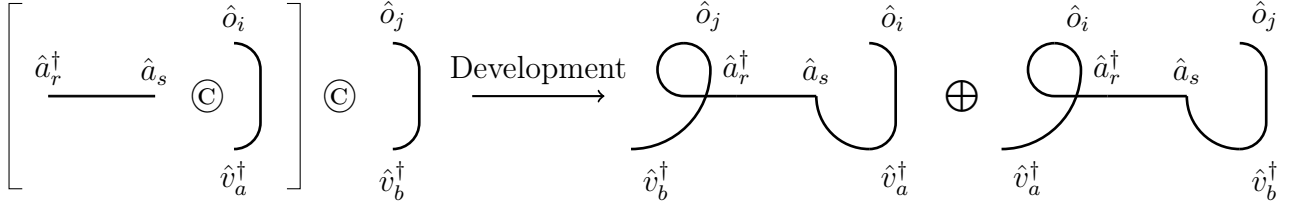
\begin{figure}[h!]
\begin{center}
\begin{tikzpicture}[scale=0.7]
\draw (-1.5,1.75) -- ++(-.25,0) -- ++(0,-3.5) -- ++(.25,0);
\draw (2,0) node {$\copyright$};
\padaga{0}{}{r}{}{s}\pexci{3}{i}{a} 
\draw (3.5,1.75) -- ++(.25,0) -- ++(0,-3.5) -- ++(-.25,0);
\tikzset{shift={(3,0)}}
\draw (1.5,0) node {$\copyright$};
\tikzset{shift={(2,0)}}
\pexci{1}{j}{b} 
\draw[->,thick] (2,0) --node[above] {Development} ++(2.5,0);
\tikzset{shift={(5,0)}}
\exci{j}{b}\adaga{deex}{r}{deex}{s}\exciend{i}{a} 
\tikzset{shift={(2.5,0)}}
\draw (-.75,0) node {$\bigoplus$};
\exci{i}{a}\adaga{deex}{r}{deex}{s}\exciend{j}{b}
\tikzset{shift={(-10,-2.5)}}
\end{tikzpicture}
\end{center}
\caption{Diagrammatic development of the $[ [ \hat{a}^\dag_r\hat{a}_s, \hat{E}_i^a ], \hat{E}_j^b ]$ nested commutators, with the inner commutator being already developed.}
\label{fig:commu_A-E-E_dev}
\end{figure}
\end{example}

\begin{example}\label{ex:commu_deex_exci}

Let $(i,j)$ be a couple of integers, both belonging to $\llbracket 1,N\rrbracket$. 
Let $(a,b)$ be a couple of integers, both belonging to $\llbracket N+1,L\rrbracket$. 
Let $\hat{a}^\dag_r$ and $\hat{a}_s$ be two second quantization operators, each pointing at a spin-orbital with non-definite occupation number relatively to the Fermi vacuum. Consider the 
\begin{equation}\label{eq:commu_D-A-E}
\left[ \left[ \hat{D}_a^i, \hat{a}^\dag_r\hat{a}_s \right], \hat{E}_j^b \right]
\end{equation}
nested commutators. The diagrammatic representation of this commutator is given in {\normalfont Figure \ref{fig:double_commu}}. The inner commutator, $[\hat{D}_a^i, \hat{a}^\dag_s\hat{a}_r ]$, has been already developed in {\normalfont Example \ref{ex:commu_D-A}}. The diagram development obtained in {\normalfont Figure \ref{fig:commu_D-A_dev}} is used to build the diagrammatic representation of the outer commutator. It is illustrated in {\normalfont Figure \ref{fig:commu_D-A-simpli-E}}. After the simplification of the inner commutator, the outer one can be considered. 

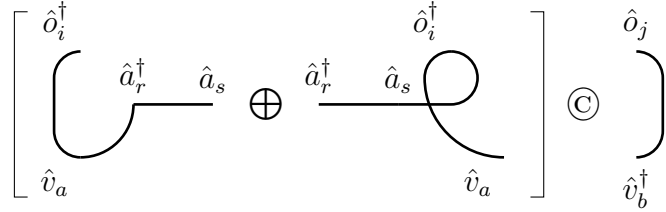
\begin{figure}[h!]
\begin{center}
\begin{tikzpicture}[scale=0.7]
\draw (-1,1.75) -- ++(-.25,0) -- ++(0,-3.5) -- ++(.25,0);
\deexend{a}{i}\adaga{exci}{r}{none}{s}
\draw (1,0) node {$\bigoplus$};
\tikzset{shift={(2,0)}}
\adaga{none}{r}{exci}{s}\deex{a}{i}
\draw (.5,1.75) -- ++(.25,0) -- ++(0,-3.5) -- ++(-.25,0);
\draw (1.5,0) node {$\copyright$};
\tikzset{shift={(2,0)}}
\pexci{1}{j}{b} 
\end{tikzpicture}
\end{center}
\caption{Diagrammatic representation of the $[[\hat{D}_a^i, \hat{a}^\dag_r\hat{a}_s], \hat{E}_j^b ]$ nested commutators, with the inner commutator being already developed.}
\label{fig:commu_D-A-simpli-E}
\end{figure}
The first sub-diagram obtained from the inner commutator can form two connections, one with the lower extremity and one with the upper extremity of the $\mathcal{L}_{1,1}$ diagram representation of the excitation operator. Both connections are represented in {\normalfont Figure \ref{fig:commu_D-A-E_simpli-1}}.

\begin{figure}[h!]
\begin{center}
\begin{tikzpicture}[scale=.7]
\deexend{a}{i}\adaga{exci}{r}{none}{s}\tcr{\pexci{2}{j}{b}}
\draw (1,0) node {$\copyright$};

\tikzset{shift={(1.5,2.5)}}

\draw[->,thick] (1.5,-1) -- (3,0) --node[above,black] {Connection of} node[below] {$\hat{o}_i^\dag$ and $\hat{o}_j$} ++(4,0);
\tikzset{shift={(8,0)}}
\tcr{\exci{j}{b}}\tikzset{shift={(1,0)}}\deex{a}{i}\adaga{exci}{r}{none}{s}

\tikzset{shift={(-12.5,-5)}}

\draw[->,thick] (1.5,1) -- (3,0) --node[above,black] {Connection of} node[below] {$\hat{a}_s$ and $\hat{v}^\dag_b$} ++(4,0);
\tikzset{shift={(8.5,0)}}
\deexend{a}{i}\adaga{exci}{r}{deex}{s}\tcr{\exciend{j}{b}}
\end{tikzpicture}
\end{center}
\caption{The possible connections between the first sub-diagram from the inner commutator diagrammatic development and the $\mathcal{L}_{1,1}$ diagrammatic representation of the excitation operator (given in {\normalfont Figure \ref{fig:commu_D-A-simpli-E}}).}
\label{fig:commu_D-A-E_simpli-1}
\end{figure}
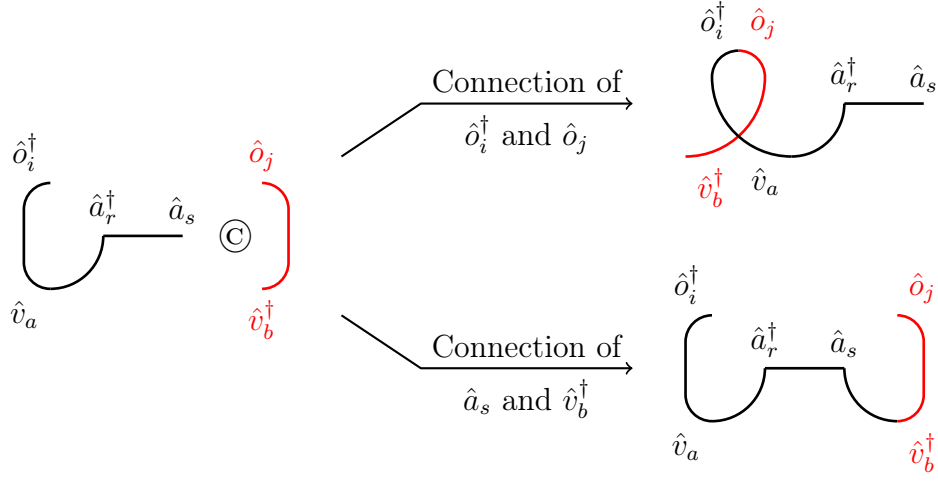

\noindent Similarly, the second sub-diagram obtained from the inner commutator can form two connections, one with the lower extremity and one with the upper extremity of the $\mathcal{L}_{1,1}$ diagram representation of the excitation operator. Both connections are represented in {\normalfont Figure \ref{fig:commu_D-A-E_simpli-2}}.

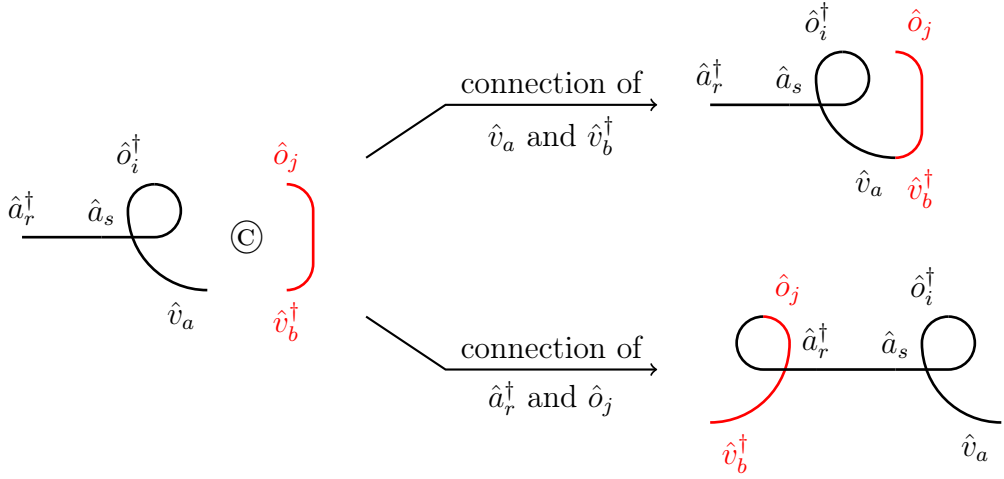
\begin{figure}[h!]
\begin{center}
\begin{tikzpicture}[scale=.7]
\adaga{none}{r}{exci}{s}\deex{a}{i}\tcr{\pexci{2}{j}{b}}
\draw (.75,0) node {$\copyright$};

\tikzset{shift={(1.5,2.5)}}

\draw[->,thick] (1.5,-1) -- (3,0) --node[above,black] {connection of}node[below] {$\hat{v}_a$ and $\hat{v}^\dag_b$} ++(4,0);
\tikzset{shift={(8,0)}}
\adaga{none}{r}{exci}{s}\deex{a}{i}\tcr{\exciend{j}{b}}

\tikzset{shift={(-11.5,-5)}}

\draw[->,thick] (1.5,1) -- (3,0) --node[above,black] {connection of}node[below] {$\hat{a}^\dag_r$ and $\hat{o}_j$} ++(4,0);
\tikzset{shift={(8,0)}}
\tcr{\exci{j}{b}}\tikzset{shift={(1,0)}}\adaga{deex}{r}{exci}{s}\deex{a}{i}
\end{tikzpicture}
\end{center}
\caption{The possible connections between the second sub-diagram from the inner commutator diagrammatic development and the $\mathcal{L}_{1,1}$ diagrammatic representation of the excitation operator (given in {\normalfont Figure \ref{fig:commu_D-A-simpli-E}}).}
\label{fig:commu_D-A-E_simpli-2}
\end{figure}

\noindent Each sub-sub-diagram derived in figures {\normalfont \ref{fig:commu_D-A-E_simpli-1}} and {\normalfont \ref{fig:commu_D-A-E_simpli-2}} can be read following {\normalfont Rule \ref{rule:interpretation_diagram_nested_commutator}} and {\normalfont Rule \ref{rule:signed_pre-computed_nested_commutator}} in order to obtain the simplification of the nested commutators:

\begin{equation}\label{eq:commu_adaga-deex-exci_result}
\left[ \left[ \hat{D}_a^i, \hat{a}^\dag_r\hat{a}_s \right], \hat{E}_j^b \right]
= \delta_{a,r}\delta_{s,b}\hat{o}^\dag_i\hat{o}_j
- \delta_{j,i}\delta_{a,r}\hat{v}^\dag_b\hat{a}_s
- \delta_{s,i}\delta_{a,b}\hat{a}^\dag_r\hat{o}_j
+ \delta_{j,r}\delta_{s,i}\hat{v}^\dag_b\hat{v}_a.
\end{equation}

\noindent This result could alternatively be derived by developing commutators and applying the anti-commutation rules.

\end{example}

\section{Algorithm}

\subsection{General description and example}

The procedure to fully expulse all brackets from a Dyck word, described in Part I, can be used to find all full contractions of a chain of second quantization operators. 
This procedure by itself generates all allowed contractions according to \eqref{eq:contraction_occ} and \eqref{eq:contraction_vir}. A method to compute the signature step by step has to be added to the procedure in order to recover the list of contractions and the corresponding signature.
The contribution of each contraction to the signature is explicit in equations \eqref{eq:contraction_in_chain-occ} and \eqref{eq:contraction_in_chain-vir}. Counting the number of operators between the two contracted ones gives the contribution to the signature. After translating the chain, the number of brackets between the two expulsed ones gives the contribution to the signature: If that number is odd (respectively, even) the sign is $(-1)$ (respectively, 1).

Let $n_{op}$ (respectively, $n_{cl}$) be the number of brackets in the last block of opening (respectively, closing) brackets. According to the definition of a Dyck word, $n_{cl}$ is greater than $n_{op}$ in a Dyck word. Among the $n_{cl}$ closing brackets, $n_{op}$ of them are chosen to be expulsed together with the last block of opening brackets.
All permutations of the set of chosen closing brackets has to be considered.
All non-chosen brackets are now part of the ``newly last'' block of closing brackets and the procedure can be repeated.

This algorithm can be direclty applied with the $\mathcal{L}_1$ translation if the expectation value is evaluated relatively to the physical vacuum. 
Relatively to the Fermi vacuum, three types of operators can appear in the chain: arbitrary, occupied, and virtual operators. The method to compute the signature can be applied for each word obtained from a split-$\mathcal{L}_2$ translation. 
First, an occupation number has to be attributed to the spin-orbital each arbitrary operator is pointing at.
Secondly, if the chain of operators is not a chain of (de)excitation operators, an additional contribution appears from the splitting of the \textit{o-}operators and \textit{v-}operators. It should be mentioned that after attributing occupation numbers of arbitrary operators and split-translating the chain, the resulting string may or may not be a Dyck word. In the latter case, the attribution of occupation numbers can be rejected. The procedure should be repeated, iterating on all the possible occupation number attributions.

\begin{example}
Let $(r_{i})_{1\leq i\leq 10}$ be a 10-tuple of integers, all belonging to $\llbracket 1,L\rrbracket$.  Consider the 
\begin{equation}
\hat{a}_{r_1}\hat{a}_{r_2}
\hat{a}_{r_3}^\dag
\hat{a}_{r_4}\hat{a}_{r_5}
\hat{a}_{r_6}^\dag
\hat{a}_{r_7}
\hat{a}_{r_8}^\dag\hat{a}_{r_9}^\dag\hat{a}_{r_{10}}^\dag
\end{equation}
chain of second quantization operators. Its $\mathcal{L}_1$-translation is represented in the initial step of {\normalfont Figure \ref{fig:block_bracket}}; blocks of closing and opening brackets are highlighted in blue and red, respectively.

\begin{figure}[h!]
\centering
\begin{tikzpicture}
\node (step_0)      at (-2,0)  
  {$\bm{\stackrel{1}{\tcr{(}}}\;
\bm{\stackrel{2}{\tcr{(}}}\;
\bm{\stackrel{3}{\tcb{)}}}\;
\bm{\stackrel{4}{\tcr{(}}}\;
\bm{\stackrel{5}{\tcr{(}}}\;
\bm{\stackrel{6}{\tcb{)}}}\;
\bm{\stackrel{7}{\tcr{(}}}\;
\bm{\stackrel{8}{\tcb{)}}}\;
\bm{\stackrel{9}{\tcb{)}}}\;
\bm{\stackrel{10}{\tcb{)}}} $};

\node (step_1)      at (5,3)  
  {$\bm{\stackrel{1}{\tcr{(}}}\;
\bm{\stackrel{2}{\tcr{(}}}\;
\bm{\stackrel{3}{\tcb{)}}}\;
\bm{\stackrel{4}{\tcr{(}}}\;
\bm{\stackrel{5}{\tcr{(}}}\;
\bm{\stackrel{6}{\tcb{)}}}\;
\bm{\stackrel{8}{\tcb{)}}}\;
\bm{\stackrel{10}{\tcb{)}}}$};

\node (step_2)    at (9,3)  
  {$\bm{\stackrel{1}{\tcr{(}}}\;
\bm{\stackrel{2}{\tcr{(}}}\;
\bm{\stackrel{3}{\tcb{)}}}\;
\bm{\stackrel{10}{\tcb{)}}}$};

\node (step_3)    at (12,3)  
  {\normalfont End};

\node (step_1')      at (5,0)  
  {$\bm{\stackrel{1}{\tcr{(}}}\;
\bm{\stackrel{2}{\tcr{(}}}\;
\bm{\stackrel{3}{\tcb{)}}}\;
\bm{\stackrel{4}{\tcr{(}}}\;
\bm{\stackrel{5}{\tcr{(}}}\;
\bm{\stackrel{6}{\tcb{)}}}\;
\bm{\stackrel{9}{\tcb{)}}}\;
\bm{\stackrel{10}{\tcb{)}}}$};

\node (step_1'')      at (5,-3)  
  {$\bm{\stackrel{1}{\tcr{(}}}\;
\bm{\stackrel{2}{\tcr{(}}}\;
\bm{\stackrel{3}{\tcb{)}}}\;
\bm{\stackrel{4}{\tcr{(}}}\;
\bm{\stackrel{5}{\tcr{(}}}\;
\bm{\stackrel{6}{\tcb{)}}}\;
\bm{\stackrel{8}{\tcb{)}}}\;
\bm{\stackrel{9}{\tcb{)}}}$};

\coordinate (S) at ($(step_0.east)!0.5!(step_1.west)$);
\coordinate (S') at  ($(step_0)!0.5!(step_1')$);
\coordinate (S'') at ($(step_0.east)!0.5!(step_1''.west)$);

\node (step_2')  at (9,0)  {$\cdots$};
\node (step_2'') at (9,-3) {$\cdots$};

\draw[->] (step_1')  to node[below] {} (step_2');
\draw[->] (step_1'') to node[below] {} (step_2'');

\node (ex1') at (5,1) {$(\tcr{7}-\tcb{8})$};
\node  at       (5,1.5)    {$1$};

\node (ex1'') at (5,-2) {$(\tcr{7}-\tcb{10})$};
\node  at        (5,-1.5)    {$1$};

\node (ex1) at (5,4)    {$(\tcr{7}-\tcb{9})$};
\node  at      (5,4.5)  {$-1$};
\node (ex2) at (9,4)    {$(\tcr{4}-\tcb{6},\tcr{5}-\tcb{8})$};
\node  at      (9,4.5)  {$-1$};
\node (ex3) at (12,4)   {$(\tcr{1}-\tcb{10},\tcr{2}-\tcb{3})$};
\node  at      (12,4.5) {$+1$};

\draw[->] (S)  	 to [bend left=30]   (ex1);
\draw[->] (S')   to [bend right=10]  (ex1');
\draw[->] (S'')  to [bend right=25]  (ex1'');

\draw[->] (step_1.east)  to [bend right=15]  (ex2.south);
\draw[->] (step_2.east)  to [bend right=15]  (ex3.south);

\draw[->] (step_0.east) to node[midway,fill=white] {Step 1a} (step_1.west);
\draw[->] (step_0.east) to node[midway,fill=white] {Step 1b} (step_1'.west);
\draw[->] (step_0.east) to node[midway,fill=white] {Step 1c} (step_1''.west);
\draw[->] (step_1)      to node[below] {Step 2a} (step_2);
\draw[->] (step_2)      to node[below] {Step 3a} (step_3);

\end{tikzpicture}
\caption{The $\bm{(\,(\,)\,(\,(\,)\,(\,)\,)\,)}$ Dyck word is an alternate succession of 3 blocks of opening brackets (in red) and 3 blocks of closing brackets (in blue). Each step corresponds to expulsion choice(s). We have only detailed the first bench of options (``Step 1'') — the Step 2 corresponding to expulsions following Step 1a also divides into branches; Step 2a is only one of the possible choices we can make after Step 1a; Step 3a is only one of the possible choices we can make after Step 2a. The contribution of each group of expulsions is specified over it. The signature is written above the expulsion details.}
\label{fig:block_bracket}
\end{figure}
\noindent In each step we expulse all bracket from the last block of opening brackets, and an identical number of closing brackets from the last block of closing brackets. One possible way to choose brackets to fully expulse the word is depicted in {\normalfont Figure \ref{fig:block_bracket}}: The {\normalfont (7-9)} pair is expulsed at {\normalfont Step 1a}; then the {\normalfont ((4-6), (5-8))} pairs are expulsed in {\normalfont Step 2a}, and at the last step the word is composed of one block of opening and one block of closing brackets. The chosen pairing in {\normalfont Step 3a} is {\normalfont ((1-10), (2-3))}.
The full contraction of the chain corresponding to this succession of choices, and its expectation value relatively to the physical vacuum are 
\begin{equation*}
\braket{
\acontraction[1ex]{}{\hat{a}}{_{r_1}\hat{a}_{r_2}\hat{a}_{r_3}^\dag\hat{a}_{r_4}\hat{a}_{r_5}\hat{a}_{r_6}^\dag\hat{a}_{r_7}\hat{a}_{r_8}^\dag\hat{a}_{r_9}^\dag}{\hat{a}}
\acontraction[2ex]{\hat{a}_{r_1}}{\hat{a}}{_{r_3}}{\hat{a}}
\acontraction[2ex]{\hat{a}_{r_1}\hat{a}_{r_2}\hat{a}_{r_3}^\dag}{\hat{a}}{_{r_4}\hat{a}_{r_5}}{\hat{a}}
\acontraction[3ex]{\hat{a}_{r_1}\hat{a}_{r_2}\hat{a}_{r_3}^\dag\hat{a}_{r_4}}{\hat{a}}{_{r_5}\hat{a}_{r_6}^\dag\hat{a}_{r_7}}{\hat{a}}
\acontraction[4ex]{\hat{a}_{r_1}\hat{a}_{r_2}\hat{a}_{r_3}^\dag\hat{a}_{r_4}\hat{a}_{r_5}\hat{a}_{r_6}^\dag}{\hat{a}}{_{r_7}\hat{a}_{r_8}^\dag}{\hat{a}}
\hat{a}_{r_1}\hat{a}_{r_2}\hat{a}_{r_3}^\dag\hat{a}_{r_4}\hat{a}_{r_5}\hat{a}_{r_6}^\dag\hat{a}_{r_7}\hat{a}_{r_8}^\dag\hat{a}_{r_9}^\dag\hat{a}_{r_{10}}^\dag
} = \delta_{r_1, r_{10}}\delta_{r_2, r_3}\delta_{r_4, r_6}\delta_{r_7, r_9}.
\end{equation*}
\noindent The signature associated with the full contraction of the chain is equal to $(-1)\times(-1)\times1 = 1$. Indeed the number of crossings between contractions is equal to $10$, leading to a signature that is $(-1)^{10} = 1$.
\end{example}

\subsection{Benchmark}

\noindent The abovementioned Dyck language-based algorithm has been implemented with Python within the \texttt{MobiDyck} code \cite{jeremy_mobidyck_2026}, which enables the symbolic evaluation of expectation values of chains of second-quantization operators relatively to the physical or Fermi vacuum. In what follows, we compare our Dyck language-based algorithm to a baseline approach based on normal ordering \cite{pdaggerq}, in which chains of second quantization operators are first normal-ordered before selecting all non-zero full contractions. For that sake, two sequences of chains are considered. 

The first sequence is $(\hat{P}_n)_{1 \leq n \leq 8}$, where each $\hat{P}_n$ consists, when read from the left to the right, of $n$ \textit{v}-annihilation operators followed by $n$ \textit{v}-creation operators. For each chain, the expectation value is evaluated relatively to the physical vacuum using both algorithms. The number of full contractions of a $\hat{P}_n$ chain, i.e., $\mathcal{W}(\hat{P}_n)$, is $n!$. Each evaluation is repeated $1000$ times, and the execution times corresponding to the total runtime over these repetitions are reported in Table \ref{tab:time_crea-anni}. For this first sequence of chains, the Dyck language-based algorithm consistently performs better than the baseline algorithm, as reported in  Table \ref{tab:time_crea-anni}. While the difference is moderate for small values of $n$, i.e., for small number of operators in the chain, it becomes substantial as the chain length increases. For instance, at $n = 8$, the Dyck language-based algorithm completes one evaluation in 0.18 seconds, whereas the normal ordering-based algorithm requires 18 seconds.

\begin{table}[h!]
\begin{center}
\begin{tabular}{c|c|cc}
$n$ & $\mathcal{W}(\hat{P}_n)$ & Normal ordering-based algorithm & Dyck language-based algorithm \\
\hline  
1 & 1 		& 0.071s 			& 0.040s \\ 
2 & 2 		& 0.090s 			& 0.068s \\ 
3 & 6 		& 0.162s 			& 0.081s \\ 
4 & 24 		& 0.702s 			& 0.137s \\ 
5 & 120 	& 5.004s 			& 0.480s \\ 
6 & 720 	& 40.887s 			& 2.820s \\ 
7 & 5040 	& 9m 37.871s 		& 20.679s \\ 
8 & 40320 	& 5h 6m 39.416s 	& 2m 57.605s \\ 
\end{tabular} 
\end{center}
\caption{Duration to symbolically evaluate $1000$ times the expectation value of $\hat{P}_n$ relatively to the physical vacuum, following two methods.}
\label{tab:time_crea-anni}
\end{table}

\noindent The second sequence is $(\hat{F}_n)_{1 \leq n \leq 5}$ where each $\hat{F}_n$ chain is composed, when read from the left to the right, of $n$ deexcitation operators followed by $n$ excitation operators. For each chain, the expectation value is evaluated with respect to the Fermi vacuum using both algorithms. The number of full contractions of an $\hat{F}_n$ chain, i.e., $\mathcal{W}(\hat{F}_n)$, is $(n!)^2$. Each evaluation is repeated $1000$ times, and the execution times corresponding to the total runtime over these repetitions are reported in Table \ref{tab:time_exci-deex}. 

\begin{table}[h!]
\begin{center}
\begin{tabular}{c|c|cc}
$n$ & $\mathcal{W}(\hat{F}_n)$ & Normal ordering-based algorithm & Dyck language-based algorithm \\
\hline  
1 & 1 		& 0.119s 		& 0.064s \\ 
2 & 4 		& 0.307s 		& 0.082s \\ 
3 & 36 		& 2.827s 		& 0.256s \\ 
4 & 576 	& 41s430 		& 2,823s \\ 
5 & 14400 	& 34m54.663s 	& 1m12.472s \\ 
%6 & 518400 	& s 			& 46m 59.868s
\end{tabular} 
\end{center}
\caption{Duration to symbolically evaluate $1000$ times the expectation value of $\hat{F}_n$ relatively to the Fermi vacuum, following two methods.}
\label{tab:time_exci-deex}
\end{table}

\newpage \noindent Reading Table \ref{tab:time_exci-deex} we see that the trend observed for the first sequence is observed for the second sequence of chains. For both algorithms, the execution time grows rapidly with the length of the chain, reflecting the combinatorial increase in the number of full contractions. However, the growth rate is significantly lower for the Dyck language-based algorithm. This improvement can be attributed to the fact that the Dyck language-based algorithm avoids the explicit construction of all full contractions of the chain. Instead, it directly exploits the structure of valid contractions through the relationship between expulsion in Dyck word and contraction in chain of second quantization operators.

Moreover, the two algorithms differ significantly in their memory requirements. The normal ordering-based algorithm requires to generate and store intermediate terms before selecting the non-zero contributions. For instance, in the case of the $\hat{P}_8$ chain, it must handle at least $40320$ non-zero contraction terms, in addition to a large number of intermediate terms that ultimately vanish. In contrast, the Dyck-language based algorithm generates only one non-zero term at a time and does not require storing the full set of contractions simultaneously. This sequential generation strategy results in a much lower memory requirement.

\section{Connection with the Goldstone diagrams}
A well-formed fully connected diagram can be uniquely associated with a graph, \textit{i.e.}, a collection of vertices connected by possibly oriented edges.

\begin{ourrule}[Construction of the graph associated with a given well-formed fully connected $\mathcal{L}_{1,1}$ diagram]\label{rule:graph_construction}
A vertex is placed at the middle of each brackets and dashes in the diagram. 
Two vertices are linked if their corresponding bracket or dash are connected. 
The resulting edge is oriented from the left to the right (respectively, right to left) vertex if the corresponding connection is in the upper (respectively, lower) half-plane. The graph obtained is then flipped clock-wise by ninety degrees.
\end{ourrule}

\noindent In the resulting graph each vertex is connected with two edges, one coming in and one going out. Every line corresponding to a connection between \textit{o-}operators are going down and every line corresponding to a connection between \textit{v-}operators are going up.
There exists a correspondence between the resulting graph and a Goldstone diagram: We have the same vertices, connected by edges with the same directionality, and we have the same labellings. This correspondence is illustrated in the following example.

\begin{example}
Let $(i,j)$ be a couple of integers, both belonging to $\llbracket 1,N\rrbracket$. 
Let $(a,b)$ be a couple of integers, both belonging to $\llbracket N+1,L\rrbracket$.
Let $\hat{a}^\dag_r$ and $\hat{a}_s$ be two second quantization operators, each pointing at a spin-orbital with non-definite occupation number relatively to the Fermi vacuum.
Consider the two following full contractions of the $\hat{D}_a^i\hat{a}^\dag_r\hat{a}_s\hat{E}_j^b$ chain:
\begin{equation}
\braket{\hat{C}_a}_{\Psi_0}
\coloneqq \braket{
\acontraction[1ex]{}{\hat{o}}{^\dag_i\hat{v}_a\hat{a}^\dag_r\hat{a}_s\hat{v}^\dag_b}{\hat{o}}
\acontraction[2ex]{\hat{o}^\dag_i}{\hat{v}}{_a}{\hat{a}}
\acontraction[2ex]{\hat{o}^\dag_i\hat{v}_a\hat{a}^\dag_r}{\hat{a}}{_s}{\hat{v}}
\hat{o}^\dag_i\hat{v}_a\hat{a}^\dag_r\hat{a}_s\hat{v}^\dag_b\hat{o}_j
}_{\Psi_0},
\quad
\braket{\hat{C}_b}_{\Psi_0}
\coloneqq \braket{
\acontraction[1ex]{}{\hat{o}}{^\dag_i\hat{v}_a\hat{a}^\dag_r}{\hat{a}}
\acontraction[2ex]{\hat{o}^\dag_i\hat{v}_a}{\hat{a}}{^\dag_r\hat{a}_s\hat{v}^\dag_b}{\hat{o}}
\acontraction[3ex]{\hat{o}^\dag_i}{\hat{v}}{_a\hat{a}^\dag_r\hat{a}_s}{\hat{v}}
\hat{o}^\dag_i\hat{v}_a\hat{a}^\dag_r\hat{a}_s\hat{v}^\dag_b\hat{o}_j
}_{\Psi_0}.
\label{eq:CaCbGoldstone}
\end{equation}

\noindent The well-formed fully connected diagram corresponding to $\braket{\hat{C}_a}_{\Psi_0}$ and $\braket{\hat{C}_b}_{\Psi_0}$ in \eqref{eq:CaCbGoldstone} are respectively diagram {\normalfont a)} and diagram {\normalfont b)} in {\normalfont Figure \ref{fig:contraction_deex-adaga-exci_decomposition}}. 
According to {\normalfont Rule \ref{rule:graph_construction}}, these two diagrams are associated with graph {\normalfont a$_1$)} and graph {\normalfont b$_1$)} represented in {\normalfont Figure \ref{fig:goldstone}}.
These two graphs are equivalent to the usual Goldstone diagrams — {\normalfont a$_2$)} and {\normalfont b$_2$)} — represented in {\normalfont Figure \ref{fig:goldstone}}.

\begin{figure}[h!]
\center
\begin{tikzpicture}[scale = 2]
%Cas a traduction
\draw (-.5,1) node {a$_1$)};
\node (top)    at (0,1) {•};
\node (center) at (0,0) {•};
\node (bottom) at (0,-1){•};

\draw [postaction={decorate}, 
decoration={markings, mark= at position 0.5 with {\arrow[line width=.75mm]{stealth}}}] 
(top.center)   node[right] {i}  to[bend left=40]  (bottom.center)node[right] {j};
\draw [postaction={decorate}, 
decoration={markings, mark= at position 0.5 with {\arrow[line width=.75mm]{stealth}}}] 
(bottom.center)node[left]  {b}  to[bend left=40]  (center.center)node[below=2mm] {s};
\draw [postaction={decorate}, 
decoration={markings, mark= at position 0.5 with {\arrow[line width=.75mm]{stealth}}}] 
(center.center)node[above=2mm] {r}  to[bend left=40]  (top.center)   node[left] {a};

%Cas a Goldstone
\draw (1,1) node {a$_2$)};
\node (top)    at (1.5,1) {•};
\node (bottom) at (1.5,-1){•};

\draw[white] (bottom.center)  to[bend left=40] node[black, above left] {r} (top.center);
\draw[white] (bottom.center)  to[bend left=40] node[black, below left] {s} (top.center);

\draw [postaction={decorate}, 
decoration={markings, mark= at position 0.5 with {\arrow[line width=.75mm]{stealth}}}] 
(top.center)   node[right] {i}  to[bend left=40]  (bottom.center)node[right] {j};
\draw [postaction=decorate,
decoration={markings, 
mark= at position 0.25 with {\arrow[line width=.75mm]{stealth}}, 
mark= at position 0.75 with {\arrow[line width=.75mm]{stealth}}}]
(bottom.center)node[left] {b}  to[bend left=40] node {•} (top.center)   node[left] {a};

%Cas b trad
\draw (3.5,1) node {b$_1$)};
\node (top)    at (4,1) {•};
\node (center) at (4,0) {•};
\node (bottom) at (4,-1){•};

\draw [postaction={decorate}, 
decoration={markings, mark= at position 0.5 with {\arrow[line width=.75mm]{stealth}}}] 
(top.center)   node[right] {i} to[bend left=40] (center.center)node[above=2mm] {s};
\draw [postaction={decorate},
decoration={markings, mark= at position 0.5 with {\arrow[line width=.75mm]{stealth}}}] 
(center.center)node[below=2mm] {r} to[bend left=40] (bottom.center)node[right] {j};
\draw [postaction={decorate},
decoration={markings, mark= at position 0.5 with {\arrow[line width=.75mm]{stealth}}}] 
(bottom.center)node[left]  {b} to[bend left=40] (top.center)   node[left]  {a};

%Cas b Goldstone
\draw (5,1) node {b$_2$)};
\node (top)    at (5.5,1) {•};
\node (bottom) at (5.5,-1){•};

\draw[white] (top.center)  to[bend left=40] node[black, above right] {s} (bottom.center);
\draw[white] (top.center)  to[bend left=40] node[black, below right] {r} (bottom.center);

\draw [postaction=decorate,
decoration={markings, 
mark= at position 0.25 with {\arrow[line width=.75mm]{stealth}}, 
mark= at position 0.75 with {\arrow[line width=.75mm]{stealth}}}]  
(top.center)   node[right] {i} to[bend left=40] node {•} (bottom.center)node[right] {j};
\draw 
[postaction=decorate,decoration={markings, mark= at position 0.5 with {\arrow[line width=.75mm]{stealth}}}]  
(bottom.center)node[left] {b} to[bend left=40]  (top.center)   node[left] {a};
\end{tikzpicture}
\caption{Comparison of the well-formed fully connected diagrams corresponding to $\braket{\hat{C}_a}_{\Psi_0}$ (a$_1$) and $\braket{\hat{C}_b}_{\Psi_0}$ (b$_1$) from \eqref{eq:CaCbGoldstone} and Goldstone diagrams corresponding to $\braket{\hat{C}_a}_{\Psi_0}$ (a$_2$) and $\braket{\hat{C}_b}_{\Psi_0}$ (b$_2$) from \eqref{eq:CaCbGoldstone}.}
\label{fig:goldstone}
\end{figure}
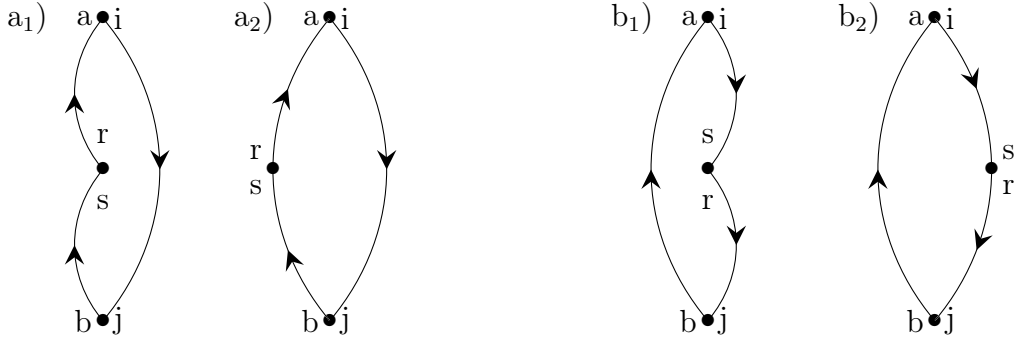
\end{example}

\begin{remark}
The main difference between Goldstone diagrams and the diagrammatic supplementation of the $\mathcal{L}_{1,1}$ alphabet lies in the method used to determine the signature associated with the diagram.
For a Goldstone diagram, the signature associated to the diagram is equal to $(-1)^{\mathcal{N}_p+\mathcal{N}_h}$ with  $\mathcal{N}_p$ the number of edges oriented from the top to the bottom, and $\mathcal{N}_h$ the number of loops.
As mentioned in {\normalfont Proposition \ref{prop:read_diag_full_contracted}}, counting the number of crossings in a well-formed fully connected $\mathcal{L}_{1,1}$ diagram is needed to obtain the associated signature.
\end{remark}

\section{Conclusion}

In this work, we have extended the connection between fermionic second quantization and Dyck language established in Part I. While translating a chain of second quantization operators into a string of brackets, information about operator indices is lost. This information has been reintroduced by supplementing the translation into a two-dimensional diagrammatic representation.

In the case of chains containing excitation operators, deexcitation operators, and blocks of creation operators followed by annihilation operators, we have shown that a full contraction of the chain can be represented as a well-formed diagram, whose connections directly encode both the contractions and the signature of the corresponding full contraction. In particular, the expectation value of a fully contracted term can be expressed as a signed product of Kronecker deltas, each being represented by a connection, while its sign is determined by the number of crossings in the associated diagram.
By introducing the depth descriptor from Dyck language to the second quantization formalism, we provided formulas to count all non-vanishing contributions in the symbolic evaluation of the expectation value of a second-quantization chain using Wick's theorem, without resorting to explicit operator manipulations. 
The interpretation of chains enabled by the introduction of depth encouraged us to implement this method in a Python program. A comparison of the initial results with a method based on normal-ordering reveals a significant improvement in performance, both in terms of execution duration and memory requirements.

From the exact same diagrammatic representation introduced to cover the corollary of Wick's theorem, a different set of manipulation rules allows for simplifying nested commutators. From the inner commutator to the outer one, diagrams are built step by step, solely selecting terms that will not cancel in the development.

Finally, a connection with Goldstone diagrams has finally been brought in our text.

\section*{Acknowledgements}

The authors are indebted to Emmanuel Giner for decisive discussions and advices. We would also like to warmly thank Julien Toulouse, Saad Yalouz, Anthony Scemama, Pierre-François Loos, and colleagues from the LPCT for very fruitful discussions on the topic.

\section*{Appendix — Proof of Proposition \ref{prop:numbering_wick}}\label{ap:N_body}

\begin{proof} We would like to remind that Remark \ref{rem:N_not_number_of_electrons} still holds during all this proof. Before all we first need to prove the following lemma:
\begin{lemma}\label{lemma:numbering_wick}
Let $\hat{C}$ be a chain of $n$ excitation and $n$ deexcitation operators. 
Let $S$ be the Fermi-$\mathcal{L}_1$ translation of $\hat{C}$.
We $N$ arbitrary creation operators followed by $N$ arbitrary annihilation operators in $\hat{C}$ — read from the left to the right — after the $K^\text{th}$ operator. The chain obtained after this introduction is named $\hat{C}^N_{K}$.
If the Fermi-$\mathcal{L}_1$ translation of $\hat{C}$ is not a Dyck word, then $\braket{\hat{C}^N_K}_{\Psi_0}$ is equal to zero.
\end{lemma}

\begin{proof} 
When applying Wick's theorem on $\hat{C}^N_{K}$, the occupation number of the spin-orbital each arbitrary operator is pointing at is constrained due to the contraction. Indeed, a contraction will not lead to an expectation value of zero if the contracted operators are both pointing at occupied spinorbitals, or are both pointing at virtual spinorbitals. In order to find all possible full contractions of $\hat{C}^N_{K}$, we can start by attributing an occupation number at the spinorbital each arbitrary operator are pointing at, and study all possible attributions.

The $\hat{C}$ chain contains the same number of \textit{o-}creation and \textit{o-}annihilation operators and the same number of \textit{v-}creation and \textit{v-}annihilation operators. It is imposed that for each arbitrary creation operator chosen to be an \textit{o-}(respectively, a \textit{v-})creation an arbitrary annihilation operator has to be chosen to be an \textit{o-}(respectively, \textit{v-})annihilation, in order to maintain the balance between creation and annihilation operators. Otherwise, the chain will not split-$\mathcal{L}_2$ translate into a pair of Dyck words. 
Among the $N$ arbitrary creation (respectively annihilation) operators $N_v$ operators are chosen to be a \textit{v-}creation (respectively \textit{v-}annihilation) operator, the $N-N_v$ remaining are chosen to be an \textit{o-}creation (respectively \textit{o-}annihilation). After this attribution, the chain of $N$ creation and $N$ annihilation operators can be $\mathcal{L}_2$-split translated: 
\begin{equation*}
\underbrace{\hat{a}^\dag_{\cdot}\cdots\hat{a}^\dag_{\cdot}}_{N}
\underbrace{\hat{a}_{\cdot}^{\textcolor{white}{\dag}}\cdots\hat{a}_{\cdot}^{\textcolor{white}{\dag}}}_{N}
\longrightarrow_{\mathcal{L}_2,\mathcal{S}} 
\underbrace{\bm{(}\cdots\bm{(}}_{N-N_v}\;\underbrace{\bm{)}\cdots\bm{)}}_{N-N_v} \;\&\; 
\underbrace{\bm{)}\cdots\bm{)}}_{N_v}\;\underbrace{\bm{(}\cdots\bm{(}}_{N_v}
\end{equation*}
Let $S_o \& S_v$ be the $\mathcal{L}_2$-split  translation of $\hat{C}$. Adding the $\bm{(}\cdots\bm{)}$ pattern to $S_o$ after the $K^\text{th}$ position cannot transform it into a Dyck word. Similarly, adding the $\bm{)}\cdots\bm{(}$ pattern to $S_v$ after the $K^\text{th}$ position cannot transform it to a Dyck word.
\end{proof}

\noindent We are now equipped for proving Proposition \ref{prop:numbering_wick}, which is stated in the conditions of Lemma \ref{lemma:numbering_wick}.

$\;$

\noindent \textit{Proof of} (i) — If $S$ is not a Dyck word, according to Lemma \ref{lemma:numbering_wick}, $\braket{\hat{C}_{K}^N}_{\Psi_0}$ is equal to zero, and $\mathcal{W}(\hat{C}_K^N)$ is equal to zero.

$\;$

\noindent \textit{Proof of} (ii) — We now consider the case in which $S$ is a Dyck word. Let $(d_k)_{0\leq k\leq 2n}$ be the list of depths of $S$ and $(d^\uparrow_k)_{1\leq k\leq n}$ the list of opening depths of $S$. In particular, $d_K$ denotes the depth reached after the $K^\text{th}$ position in $S$. Let $S'_o \& S'_v$ be the $\mathcal{L}_2$-split translation of $\hat{C}_{K}^N$.
The list of depths of $S'_v$ (respectively $S'_o$) contains at least all elements of the list of depths of $S_v$ (repectively $S_o$). We have the three following cases:

$\;$\\
(\textit{Case} ii.a) $0<N_v<N$. The list of depths of $S'_v$ includes the list of opening depths of $S_v$ and the elements of $({d_K}-N_v+1, \ldots, {d_K})$. Indeed, each closing bracket decreases the depth by one from ${d_K}$ to $({d_K}-N_v)$. The $N_v$ opening brackets bring the depth back at a value of $d_K$, hence the $k^\text{th}$ bracket adds $({d_K}-N_v+k)$ to the list of opening depths of $S'_v$.  
The list of depths of $S'_o$ includes the list of opening depths of $S_o$ and the elements of $({d_K}+1, \ldots, {d_K}+N-N_v)$. Indeed, starting at a depth of ${d_K}$ each opening brackets increases the depth by one, resulting in adding all integers between $({d_K}+1)$ and $({d_K}+N-N_v)$. The $(N-N_v)$ closing brackets bring the depth back at a value of ${d_K}$. 

$\;$\\
(\textit{Case} ii.b) $N_v = 0$. The $S'_o$ word is obtained by adding $N$ opening brackets followed by $N$ closing brackets to $S_o$ after the $K^\text{th}$ bracket. The list of opening depths of $S_o'$ contains the elements of the list of opening depths of $S_o$ and the elements of $({d_K}+1, \ldots, d+N)$, while the elements of the list of opening depths of $S_v'$ and $S_v$ are the same.

$\;$\\
(\textit{Case} ii.c) $N_v = N$. The $S'_v$ word is obtained by adding $N$ opening brackets followed by $N$ closing brackets to $S_v$ after the $K^\text{th}$ bracket. The list of opening depths of $S_v'$ contains the elements of the list of opening depths of $S_o$ and the elements of $({d_K}-N+1, \dots, {d_K})$, while the elements of the list of opening depths of $S_o'$ and $S_o$ are the same.

$\;$\\
For every value of $N_v$ the list of opening depths of $S_o'$ and $S_v'$ contains the same elements as the list of opening depths of $S_o'$ and $S_v'$, and the elements of $({d_K}+1-N_v, \ldots, {d_K}+N-N_v)$. According to Proposition III.5 from Part I the number of ways to expulse all brackets pairs of a Dyck word is equal to the product of its opening depths. This proposition holds for $S_o'$ and $S_v'$. The number of ways to fully contract the chain $\hat{C}_{K}^N$ is equal to the number of ways to expulse all brackets pairs from $S_o'$ times the number of ways to expulse all brackets pairs from $S_v'$, which is equal to 

\begin{equation}\label{eq-ap:number_nv}
\mathcal{W}(\hat{C}) \prod_{k=1}^N ({d_K}+k-N_v) = \mathcal{W}(\hat{C})\dfrac{({d_K}+N-N_v)!}{({d_K}-N_v)!} 
\end{equation}
If the value of $N_v$ exceeds $d_K$ a strictly negative value of depth will appear in the list of depths of $S_v'$ in Case ii.a. Hence, according to Proposition III.1 in Part I, the $S'_v$ string is not a Dyck word. The value of $N_v$ equal to or smaller that ${d_K}$.
For a given value of $N_v$, there are $\binom{N}{N_v}$ ways $N_v$ to choose arbitrary creation (respectively, annihilation) operators among the $N$ arbitrary creation (respectively, annihilation) operators in $\hat{C}_{K}^N$. The number of full contractions of $\hat{C}_{K}^N$ whose expectation value relatively to the Fermi vacuum is not equal to zero due to \eqref{eq:contraction_occ} to \eqref{eq:contraction_in_chain-vir} is therefore obtained by summing \eqref{eq-ap:number_nv} over all possible value of $N_v$, namely

\begin{equation*}
\mathcal{W}(\hat{C}_{K}^N) = \mathcal{W}(\hat{C})\sum_{N_v=0}^{min(N,{d_K})}\binom{N}{N_v}^2\dfrac{({d_K}+N-N_v)!}{(d-N_v)!}.
\end{equation*}
In order to write $\mathcal{W}(\hat{C}_{K}^N)$ from $\mathcal{W}(\hat{C})$ and $N!$, the $\mathcal{M}_{d_K}^N$ multiplicative factor is defined as
\begin{equation*}
\mathcal{M}_{d_K}^N = \dfrac{1}{N!}\sum_{N_v=0}^{min(N,{d_K})}\binom{N}{N_v}^2\dfrac{({d_K}+N-N_v)!}{({d_K}-N_v)!},
\end{equation*}
which can be written as
\begin{equation*}
\mathcal{M}_{d_K}^N = \sum_{N_v=0}^{min(N,{d_K})}\binom{N}{N_v}^2\binom{{d_K}+N-N_v}{N}.
\end{equation*}
The $N_v$ variable is a dummy variable, so it has been replaced by $\ell$ in formula \eqref{eq:multiplicateur}.
\end{proof}
\newpage

\bibliographystyle{unsrt}

\begin{thebibliography}{10}

\bibitem{helgaker_molecular_2014}
Trygve Helgaker, Poul Jørgensen, and Jeppe Olsen.
\newblock {\em Molecular {Electronic}-{Structure} {Theory}}.
\newblock Wiley, Hoboken, 2014.

\bibitem{lefebvre_use_1969}
Jiří Čížek.
\newblock On the {Use} of the {Cluster} {Expansion} and the {Technique} of
  {Diagrams} in {Calculations} of {Correlation} {Effects} in {Atoms} and
  {Molecules}.
\newblock In R.~LeFebvre and C.~Moser, editors, {\em Advances in {Chemical}
  {Physics}}, volume~14, pages 35--89. Wiley, 1 edition, January 1969.

\bibitem{cizek_correlation_1966}
Jiří Čížek.
\newblock On the {Correlation} {Problem} in {Atomic} and {Molecular} {Systems}.
  {Calculation} of {Wavefunction} {Components} in {Ursell}-{Type} {Expansion}
  {Using} {Quantum}-{Field} {Theoretical} {Methods}.
\newblock {\em The Journal of Chemical Physics}, 45(11):4256--4266, December
  1966.

\bibitem{cizek_correlation_1971}
J.~Čižek and J.~Paldus.
\newblock Correlation problems in atomic and molecular systems {III}.
  {Rederivation} of the coupled‐pair many‐electron theory using the
  traditional quantum chemical methodst.
\newblock {\em International Journal of Quantum Chemistry}, 5(4):359--379, July
  1971.

\bibitem{shavitt_many-body_2009}
Isaiah Shavitt and Rodney~J. Bartlett.
\newblock {\em Many-body methods in chemistry and physics: {MBPT} and
  coupled-cluster theory}.
\newblock Cambridge molecular science. Cambridge University Press, Cambridge
  New York, 2009.

\bibitem{paldus_correlation_1972}
J.~Paldus, J.~Čížek, and I.~Shavitt.
\newblock Correlation {Problems} in {Atomic} and {Molecular} {Systems}. {IV}.
  {Extended} {Coupled}-{Pair} {Many}-{Electron} {Theory} and {Its}
  {Application} to the {B} {H} 3 {Molecule}.
\newblock {\em Physical Review A}, 5(1):50--67, January 1972.

\bibitem{evangelista_adaptive_2014}
Francesco~A. Evangelista.
\newblock Adaptive multiconfigurational wave functions.
\newblock {\em The Journal of Chemical Physics}, 140(12):124114, March 2014.
\newblock arXiv:1403.4117 [physics].

\bibitem{lyakh_adaptive_2010}
Dmitry~I. Lyakh and Rodney~J. Bartlett.
\newblock An adaptive coupled-cluster theory: @{CC} approach.
\newblock {\em The Journal of Chemical Physics}, 133(24):244112, December 2010.

\bibitem{wick_evaluation_1950}
G.~C. Wick.
\newblock The {Evaluation} of the {Collision} {Matrix}.
\newblock {\em Physical Review}, 80(2):268--272, October 1950.

\bibitem{lindgren_atomic_1986}
Ingvar Lindgren and John Morrison.
\newblock {\em Atomic {Many}-{Body} {Theory}}.
\newblock Springer Berlin Heidelberg, Berlin, Heidelberg, 1986.

\bibitem{surjan_second_1989}
Péter~R. Surján.
\newblock {\em Second quantized approach to quantum chemistry: an elementary
  introduction: with 11 figures}.
\newblock Springer-Verlag, Berlin Heidelberg New York London Paris Tokyo Hong
  Kong, softcover reprint of the hardcover 1st edition 1989 edition, 1989.

\bibitem{gross_many-particle_1991}
E.~K.~U. Gross, Erich Runge, Olle Heinonen, and E.~K.~U. Gross.
\newblock {\em Many-particle theory}.
\newblock Hilger, Bristol, 1991.

\bibitem{kutzelnigg_normal_1997}
Werner Kutzelnigg and Debashis Mukherjee.
\newblock Normal order and extended {Wick} theorem for a multiconfiguration
  reference wave function.
\newblock {\em The Journal of Chemical Physics}, 107(2):432--449, July 1997.

\bibitem{wilson_methods_1992}
Stephen Wilson and Geerd H.~F. Diercksen, editors.
\newblock {\em Methods in {Computational} {Molecular} {Physics}}, volume 293 of
  {\em {NATO} {ASI} {Series}}.
\newblock Springer US, Boston, MA, 1992.

\bibitem{goldstone_derivation_1957}
J~Goldstone.
\newblock Derivation of the {Brueckner} many-body theory.
\newblock 239(1217):267--279, February 1957.

\bibitem{mattuck_guide_1992}
Richard~D. Mattuck.
\newblock {\em A guide to {Feynman} diagrams in the many-body problem}.
\newblock Dover pub, New York, 2nd ed edition, 1992.

\bibitem{szabo_modern_2012}
Attila Szabo and Neil~S. Ostlund.
\newblock {\em Modern quantum chemistry: introduction to advanced electronic
  structure theory}.
\newblock Dover Publications, Inc, Mineola, New York, 2012.

\bibitem{hugenholtz_perturbation_1957}
N.M. Hugenholtz.
\newblock Perturbation theory of large quantum systems.
\newblock {\em Physica}, 23(1-5):481--532, January 1957.

\bibitem{kucharski_fifth-order_1986}
Stanislaw~A. Kucharski and Rodney~J. Bartlett.
\newblock Fifth-{Order} {Many}-{Body} {Perturbation} {Theory} and {Its}
  {Relationship} to {Various} {Coupled}-{Cluster} {Approaches}.
\newblock In {\em Advances in {Quantum} {Chemistry}}, volume~18, pages
  281--344. Elsevier, 1986.

\bibitem{crawford_introduction_2000}
T.~Daniel Crawford and Henry~F. Schaefer.
\newblock An {Introduction} to {Coupled} {Cluster} {Theory} for {Computational}
  {Chemists}.
\newblock In Kenny~B. Lipkowitz and Donald~B. Boyd, editors, {\em Reviews in
  {Computational} {Chemistry}}, volume~14, pages 33--136. Wiley, 1 edition,
  January 2000.

\bibitem{roman_introduction_2015}
Steven Roman.
\newblock {\em An introduction to catalan numbers}.
\newblock Compact textbooks in mathematics. Birkhäuser, Cham Heidelberg New
  York Dordrecht London, softcover reprint of the hardcover 1st edition 2015
  edition, 2015.

\bibitem{conway_lowdimensional_1997}
J.~H. Conway and N.~J.~A. Sloane.
\newblock Low–dimensional lattices. {VII}. {Coordination} sequences.
\newblock {\em Proceedings of the Royal Society of London. Series A:
  Mathematical, Physical and Engineering Sciences}, 453(1966):2369--2389,
  November 1997.

\bibitem{jeremy_mobidyck_2026}
Jérémy Morere.
\newblock \texttt{MobiDyck}.
\newblock \url{https://github.com/Jeremy-Morere/MobiDyck}, 2026.
\newblock Accessed: 2026-05-25.

\bibitem{pdaggerq}
Eugene DePrince.
\newblock \texttt{pdaggerq}.
\newblock \url{https://github.com/edeprince3/pdaggerq/}, 2025.
\newblock Accessed: 2026-03-12.

\end{thebibliography}

\end{document}